%% file: main.tex
\documentclass[conference]{IEEEtran}
\IEEEoverridecommandlockouts

\usepackage{cite}
\usepackage{amsmath,amssymb,amsfonts}
\usepackage{graphicx}
\usepackage{caption}
\usepackage{textcomp}
\usepackage{xcolor}
\usepackage{ltl}
\usepackage{tikz, tikz-3dplot}
\usetikzlibrary{arrows,backgrounds,decorations,decorations.pathmorphing,
positioning,fit,automata,shapes,snakes,patterns,plotmarks,calc,trees,arrows.meta}
\usepackage{algorithm}
\usepackage{algorithmicx}
\usepackage{algpseudocode}
\usepackage{multicol}
\usepackage{listings}
\usepackage{color}
\usepackage{wrapfig}
\usepackage{subcaption}
\usepackage{xspace}
\usepackage{pgfplots}
\pgfplotsset{compat=newest}
\usepackage{pgfplotstable}
\usepackage{pifont}
\usepackage{todonotes}
\usepackage{hyperref}
\hypersetup{%
    colorlinks=true,           
    allcolors=blue!70!black,   
    pdfstartview=Fit,          
    breaklinks=true,
    pdfauthor={gxjlsbh},
    pdftitle={Distributed Runtime Verification of Metric Temporal Properties}}
 
\def\BibTeX{{\rm B\kern-.05em{\sc i\kern-.025em b}\kern-.08em
    T\kern-.1667em\lower.7ex\hbox{E}\kern-.125emX}}
\begin{document}

\renewcommand{\baselinestretch}{.96}
\title{Stream-based Decentralized Runtime Verification
}

%
\author{\IEEEauthorblockN{Ritam Ganguly}
    \IEEEauthorblockA{\textit{Department of Computer Science and Engineering} \\
        \textit{Michigan State University}\\
        East Lansing, USA \\
        gangulyr@msu.edu}
    \and
    \IEEEauthorblockN{Borzoo Bonakdarpour}
    \IEEEauthorblockA{\textit{Department of Computer Science and Engineering} \\
        \textit{Michigan State University}\\
        East Lansing, USA \\
        borzoo@msu.edu}
    }
%

\maketitle

\thispagestyle{plain}
\pagestyle{plain}

\newcommand{\MTL}{\textsf{\small MTL}\xspace}
\newcommand{\AMTL}{\textsf{\small AMTL}\xspace}
\newcommand{\LTL}{\textsf{LTL}\xspace}
\newcommand{\FLTL}{\texttt{FLTL}\xspace}
\newcommand{\lola}{\textsc{Lola}\xspace}

\newcommand{\intervalSet}{\mathbb{I}}
\newcommand{\naturalSet}{\mathbb{N}}
\newcommand{\realSet}{\mathbb{R}}
\newcommand{\realPlusSet}{\mathbb{R}_{\geq 0}}
\newcommand{\wholeSet}{\mathbb{Z}}
\newcommand{\wholePlusSet}{\mathbb{Z}_{\geq 0}}

\newcommand{\qed}{$~\blacksquare$}

\newcommand{\interval}{\mathcal{I}}
\newcommand{\Pred}{\mathsf{AP}}
\newcommand{\p}{\mathsf{p}}
\newcommand{\Time}{\tau}
\newcommand{\StrucV}{\bar{\mathcal{D}}}
\newcommand{\TimeV}{\bar{\tau}}
\newcommand{\Istart}{\mathit{start}}
\newcommand{\Iend}{\mathit{end}}

\newcommand{\tru}{\mathtt{true}}
\newcommand{\fals}{\mathtt{false}}
\newcommand{\ite}{\mathtt{ite}}
\newcommand{\val}{\mathit{val}}
\newcommand{\pr}{\texttt{T}}
\newcommand{\ab}{\texttt{F}}
\newcommand{\nul}{\natural}

\DeclareRobustCommand{\F}{\LTLdiamond}
\DeclareRobustCommand{\G}{\LTLsquare}
\DeclareRobustCommand{\U}{\,\mathcal U \,}
\DeclareRobustCommand{\X}{\LTLcircle}

\newcommand{\allVar}{\mathcal{V}}
\newcommand{\var}{v}
\newcommand{\domVar}{\mathcal{D}_v}

\newcommand{\stream}{\alpha}
\newcommand{\type}{\mathsf{T}}
\newcommand{\allType}{\mathbb{T}}

\newcommand{\allStream}{\mathcal{A}}
\newcommand{\Proc}{\mathcal{P}}
\newcommand{\Sys}{\textit{SYS}}
\newcommand{\globalC}{\mathcal{G}}
\newcommand{\Events}{\mathcal{E}}
\newcommand{\hb}{\rightsquigarrow}
\newcommand{\cc}{\mathcal{C}}
\newcommand{\ccAll}{\mathbb{C}}
\newcommand{\front}{\mathsf{front}}
\newcommand{\Sr}{\mathsf{Sr}}
\newcommand{\seg}{\textit{seg}}
\newcommand{\trace}{\alpha}
\newcommand{\RTime}{\sigma}
\newcommand{\RTimeV}{\bar{\sigma}}
\newcommand{\LTime}{\pi}
\newcommand{\LTimeV}{\bar{\pi}}

\newcommand{\Monitors}{\mathcal{M}}
\newcommand{\LS}{\mathit{LS}}
\newcommand{\LC}{\mathit{LC}}

\newcommand{\Pro}{\mathsf{Pr}}
\newcommand{\hbSet}{\mathsf{hbSet}}
\newcommand{\cond}{\mathsf{cond}}
\newcommand{\Cond}{\mathbf{COND}}
\newcommand{\Prob}{\mathcal Pr}
\newcommand{\Result}{\mathsf{Result}}

\newtheorem{definition}{Definition}
\newcommand{\InInt}[1]{\mathsf{InInt(#1)}}

\algrenewcommand\algorithmicindent{.5em}%

\newcommand{\code}[1]{\textsf{\small #1}\xspace}
\newcommand{\spec}{\mathsf{spec}}
\newcommand{\liveness}{\mathsf{liveness}}

\algdef{SE}[DOWHILE]{Do}{doWhile}{\algorithmicdo}[1]{\algorithmicwhile\ #1}

 \newtheorem{theorem}{Theorem}
 \newtheorem{lemma}{Lemma}
 \newtheorem{example}{Example}
 \newtheorem{proof}{Proof}
 
 \newcommand{\mytodo}[3]{\todo[linecolor=#1,backgroundcolor=#1!25,bordercolor=#1]{#2: #3}}
\newcommand{\myinline}[3]{\todo[inline,linecolor=#1,backgroundcolor=#1!25,bordercolor=#1]{#2: 
        #3}}

\newcommand{\borzootodo}[1]{\mytodo{red}{Borzoo}{#1}}
\newcommand{\borzooinline}[1]{\myinline{red}{Borzoo}{#1}}

\newcommand{\ritamtodo}[1]{\mytodo{green}{Ritam}{#1}}
\newcommand{\ritaminline}[1]{\myinline{green}{Ritam}{#1}}

\input{abs}
\input{intro}
\input{prelim}

\input{problem}
\input{progression}
\input{smt}
\input{monitor}
\input{eval}

\section{Related Work}
\label{sec:related}

Online predicate detection for both centralized and decentralized monitoring setting have been 
extensively studies in~\cite{cgnm13,mg05}. Extensions to more expressive temporal operators are 
introduced in~\cite{og07,mb15}. Monitoring approaches introduced 
in~\cite{cgnm13,og07,mb15} considers a fully asynchronous distributed system. An 
SMT-based predicate detection solution has been introduced in~\cite{vyktd17}. Runtime Verification 
for {\em synchronous} distributed system has been studied in~\cite{ds19,cf16,bf16}. The assumption 
of a common global clock shared among all the components act as a major shortcoming of this 
approach. Finally, fault-tolerant monitoring, where monitors can crash, has been investigated 
in~\cite{bfrrt16} for asynchronous and in~\cite{kb18} for synchronized distributed processes.

Runtime Verification of stream-based specification was introduced in~\cite{lola05, chlsst18}, where 
the occurrence of the events was assumed to be synchronous. To extend the stream-based runtime 
verification to more complex systems, one where the occurrence of events is asynchronous,
a real-time based logic was introduced in~\cite{t19, lssst19, lssss20}. However, these methods fall 
short to verify large geographically separated distributed system, due to their assumption regarding 
the presence of a shared global clock. On the contrary, we assume the presence of a clock 
synchronization algorithm which limits the maximum clock skew among components to a constant. 
This is a realistic assumption since different components of a large industrial system have their 
own clock and it is certain to have a skew between them. A similar SMT-based solution was studied 
for \LTL and \MTL specifications in~\cite{ganguly20, ganguly22} respectively, which we extend to 
include a more expressive stream-based specification.


\section{Conclusion}
\label{sec:concl}

In this paper, we studied distributed runtime verification w.r.t. to the popular stream-based specification language \lola. 
We propose a online decentralized monitoring approach where each monitor takes a set of associated 
\lola specification and a partial distributed stream as input. By assuming partial synchrony 
among all streams and by reducing the verification problem into an SMT problem, we were able to 
reduce the complexity of our approach where it is no longer dependent on the time synchronization 
constant. We also conducted extensive synthetic experiments, verified system properties of large 
Industrial Control Systems and airspace monitoring of SBS messages. Comparing to machine learning-based approaches to verify the 
correctness of these system, our approach was able to produce sound and correct results with 
deterministic guarantees. As a better practice, one can also use our RV approach 
along with machine-learning based during training or as a safety net when detecting system violations.

For future work, we plan to study monitoring of distributed systems where monitors themselves are vulnerable 
to faults such as crash and Byzantine faults. This will let us design a technique with faults and 
vulnerabilities mimicking a real life monitoring system and thereby expanding the reach and 
application of runtime verification on more real-life safety critical systems.


\bibliographystyle{IEEEtran}
\bibliography{bibliography}

\newpage
\input{appendix}

\end{document}

%% file: abs.tex
\begin{abstract}

Industrial Control Systems (ICS) are often built from  geographically distributed components and often use programmable logic controllers for localized processes. 
Since verification of such systems is challenging because of both time sensitivity of the 
system specifications and the inherent asynchrony in distributed components, developing runtime assurance that verifies not just the correctness of different components, but 
also generates aggregated statistics of the systems is of interest. 
In this paper, we first present a general technique for runtime monitoring of distributed applications 
whose behavior can be modeled as input/output {\em streams} with an internal computation module 
in the partially synchronous semantics, where an imperfect clock synchronization algorithm is 
assumed.
Second, we propose a generalized stream-based decentralized runtime verification technique.
We also rigorously evaluate our algorithm on extensive synthetic experiments and several ICS and aircraft SBS message datasets.
\end{abstract}


%% file: intro.tex
\vspace{-1mm}
\section{Introduction}
\label{sec:intro}

Industrial Control Systems (ICS) are information systems to control 
industrial processes such as 
manufacturing, product handling, distribution, etc. It includes supervisory control and data 
acquisition systems used to control geographically dispersed assets and distributed control systems 
using a programmable logic controller for each of the localized processes. A typical programmable 
logic controller (PLC) receives data produced by a large number of sensors, fitted across the 
system. The data produced by these components are often the target of 
cyber and ransom-ware attack 
putting the security of the system in jeopardy. Since these systems are 
linked to essential 
services, any attack on these facilities put the users life on the front line. The integrity of 
the data produced from these distributed components is very important as the PLC's behavior is 
dictated by it. Recent attacks have shown that an attack on a company's ICS costs the company 
around \$5 million and 50 days of system down time. Additionally, 
according to a recent report~\cite{icsReport}, it takes the effected 
company around 191 days to fully recover and around 
54\% of all organization are vulnerable to such attacks.

In this paper, we advocate for a runtime verification (RV) approach, to monitor the behavior 
of a distributed system with respect to a formal specification. Applying RV 
to multiple 
components of an ICS can be viewed as the general problem of distributed 
RV, where a centralized or decentralized 
monitor(s) observe the behavior of a distributed system in which the processes do not share a 
global clock. Although RV deals with finite executions, the lack of a common global clock 
prohibits it from having a total ordering of events in a distributed setting. In other words, 
the monitor can only form a partial ordering of events which may yield different evaluations. 
Enumerating all possible interleavings of the system at runtime incurs in an exponential blowup, 
making the approach not scalable. To add to this already complex task, a PLC often requires 
time sensitive aggregation of data from multiple sources.

We propose an effective, sound and complete solution to distributed RV for the popular 
{\em stream-based} specification language \lola~\cite{lola05}. 
Compared to other temporal 
logic, \lola can describe both correctness/failure assertions along with statistical measures 
that can be used for system profiling and coverage analysis.
To present a high level of \lola example, consider two input 
streams $x$ and $y$ and a output stream, $\mathit{sum}$ as shown in 
Fig.~\ref{fig:intro1}. 
Stream $x$ has the value $3$ until time instance $2$ when it changes to 
$5$ and so on.

\begin{lstlisting}
input x:int
input y:int
output sum := x+y
\end{lstlisting}

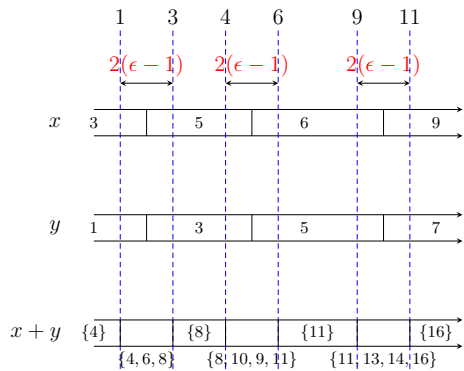
\begin{figure}
\centering
\scalebox{.7}{
	\input{fig_intro}
}
	\caption{Partially Synchronous LOLA}
	\label{fig:intro1}
	\vspace{-6mm}
\end{figure}

We consider a fault proof decentralized set of monitors where each monitor only has a partial view 
of the system and has no access to a global clock. In order to limit the blow-up of states posed by 
the absence of the global clock, we make a practical assumption about the presence of a bounded 
clock skew $\epsilon$ between all the local clocks, guaranteed by a clock 
synchronization algorithm (like NTP~\cite{ntp}). This setting is known to be 
{\em partially synchronous}. 
As can be seen in Fig.~\ref{fig:intro1}, any two events less than $\epsilon = 2$ time apart 
is considered to be concurrent and thus the non-determinism of the time of occurrence of each event is restricted to $\epsilon - 1$ on either side. When attempting to evaluate the output stream 
$\mathit{sum}$, 
we need to take into consideration all the possible time of occurrence of the 
values. For example, when 
evaluating the value of $\mathit{sum}$ at time $1$, we need to consider the 
value of $x$ (resp. $y$) as $3$ and 
$5$ (resp. $1$ and $3$) which evaluates to $4$, $6$ and $8$. The same can be observed for 
evaluations across all time instances.

Our first contribution in this paper is introducing a partially synchronous 
semantics for \lola. 
In other words, we define \lola which takes into consideration a clock-skew of $\epsilon$ when 
evaluating a stream expression. Second, we introduce an SMT-based 
associated equation rewriting 
technique over a partially observable distributed system, which takes into 
consideration the 
values observed by the monitor and rewrites the associated equation. The 
monitors are able to 
communicate within themselves and are able to resolve the partially evaluated equations into 
completely evaluated ones.

We have proved the correctness of our approach and the upper and lower 
bound of the message complexity. Additionally, we have completely implemented 
our technique and report the results of rigorous synthetic experiments, as well as monitoring 
correctness and aggregated results of several ICS. 
As identified in~\cite{acz20}, most attacks on ICS components try to alter the 
value reported to the PLC in-order to make the PLC behave erroneously. Through our approach, we 
were able to detect these attacks in-spite of the clock asynchrony among the different components 
with deterministic guarantee. We also argue that our approach was able to evaluate system behavior 
aggregates that makes studying these system easier by the human operator. Unlike machine learning approaches (e.g., ~\cite{pma15, pma15a, bbbmap14}), our approach will never raise false negatives.
We put our monitoring technique to test, studying 
the effects of different parameters on the runtime and size of the message sent from one monitor to 
other and report on each of them.

\paragraph*{Organization} Section~\ref{sec:prelim} presents the background concepts. Partially synchronous \lola and the formal problem statement are introduced in Section~\ref{sec:pslola}.
Our RV technique is collectively presented in Sections
Section~\ref{sec:problem} -- \ref{sec:mon} followed by the 
experimental results in Section~\ref{sec:eval}. Related work is discussed in Section~\ref{sec:related} 
before we make concluding remarks in Section~\ref{sec:concl}. Details of syntax of \lola, proofs of correctness and more details about the ICS case studies can be found in the Appendix~\ref{sec:appendix}.

%% file: fig_intro.tex
    \begin{tikzpicture}
        \tikzstyle{every node}=[font=\large]
        \tikzset{
            >=stealth,
            every state/.style={thick, fill=gray!10},
        }
        
        \draw (0,5) node[above] {\textcolor{white}{.}};
        
        \draw (-0.5,4) node[left] {$x$};
        
        \draw (-0.5,2) node[left] {$y$};

        \draw (-0.5,0) node[left] {$x+y$};
        
        \draw [->] (0,4.25) -- (7,4.25);
        \draw [->] (0,3.75) -- (7,3.75);
        
        \draw [->] (0,2.25) -- (7,2.25);
        \draw [->] (0,1.75) -- (7,1.75);

        \draw [->] (0,0.25) -- (7,0.25);
        \draw [->] (0,-0.25) -- (7,-0.25);
        
        \draw [black,fill=gray!10] (0,4) node[] {\small $3$};
        \draw[-] (1,4.25) -- (1,3.75);
        \draw [black,fill=gray!10] (2,4) node[] {\small $5$};
        \draw[-] (3,4.25) -- (3,3.75);
        \draw [black,fill=gray!10] (4,4) node[] {\small $6$};
        \draw[-] (5.5,4.25) -- (5.5,3.75);
        \draw [black,fill=gray!10] (6.5,4) node[] {\small $9$};
        
        \draw [black,fill=gray!10] (0,2) node[] {\small $1$};
        \draw[-] (1,2.25) -- (1,1.75);
        \draw [black,fill=gray!10] (2,2) node[] {\small $3$};
        \draw[-] (3,2.25) -- (3,1.75);
        \draw [black,fill=gray!10] (4,2) node[] {\small $5$};
        \draw[-] (5.5,2.25) -- (5.5,1.75);
        \draw [black,fill=gray!10] (6.5,2) node[] {\small $7$};

        \draw [blue, densely dashed] (0.5,5.75) -- (0.5,-0.75);
        \draw[<->] (0.5,4.75) -- (1.5,4.75) node[above, pos=0.5] {\textcolor{red}{$2(\epsilon-1)$}};
        \draw [blue, densely dashed] (1.5,5.75) -- (1.5,-0.75);
        \draw [blue, densely dashed] (2.5,5.75) -- (2.5,-0.75);
        \draw[<->] (2.5,4.75) -- (3.5,4.75) node[above, pos=0.5] {\textcolor{red}{$2(\epsilon-1)$}};
        \draw [blue, densely dashed] (3.5,5.75) -- (3.5,-0.75);
        \draw [blue, densely dashed] (5,5.75) -- (5,-0.75);
        \draw[<->] (5,4.75) -- (6,4.75) node[above, pos=0.5] {\textcolor{red}{$2(\epsilon-1)$}};
        \draw [blue, densely dashed] (6,5.75) -- (6,-0.75);

        \draw(0.5,5.75) node[above] {$1$};
        \draw(1.5,5.75) node[above] {$3$};
        \draw(2.5,5.75) node[above] {$4$};
        \draw(3.5,5.75) node[above] {$6$};
        \draw(5,5.75) node[above] {$9$};
        \draw(6,5.75) node[above] {$11$};

        \draw (0,0) node[] {\small $\{4\}$};
        \draw[-] (0.5,0.25) -- (0.5,-0.25);
        \draw (1,-0.5) node[] {\small $\{4, 6, 8\}$};
        \draw[-] (1.5,0.25) -- (1.5,-0.25);
        \draw (2,0) node[] {\small $\{8\}$};
        \draw[-] (2.5,0.25) -- (2.5,-0.25);
        \draw (3,-0.5) node[] {\small $\{8, 10, 9, 11\}$};
        \draw[-] (3.5,0.25) -- (3.5,-0.25);
        \draw (4.25,0) node[] {\small $\{11\}$};
        \draw[-] (5,0.25) -- (5,-0.25);
        \draw (5.5,-0.5) node[] {\small $\{11, 13, 14, 16\}$};
        \draw[-] (6,0.25) -- (6,-0.25);
        \draw (6.5,0) node[] {\small $\{16\}$};
        
    \end{tikzpicture}

%% file: prelim.tex
\section{Preliminaries -- Stream-based Specification Language (\lola)~\cite{lola05}}
\label{sec:prelim}

%

A \lola~\cite{lola05} specification describes the computation of output streams given a set of 
input streams. 
A {\em stream} $\stream$ of type $\type$ is a finite sequence of values, $t \in \type$.
Let $\stream(i)$, where $i \geq 0$, denote the value of the stream at time stamp $i$. We denote a 
stream of finite length (resp. infinite length) by $\type^*$ (resp. $\type^\omega$).

\begin{definition}\label{def:lola}
A \lola specification is a set of equations over typed stream variables of the form:
\begin{align*}
s_1 &= e_1(t_1, \cdots, t_m, s_1, \cdots, s_n) \\
\vdots &~~~ \vdots \\
s_n &= e_n(t_1, \cdots, t_m, s_1, \cdots, s_n)
\end{align*}
where $s_1, s_2, \cdots, s_n$ are called the {\em dependent variables}, $t_1, t_2, \cdots, t_m$ 
are called the {\em independent variables}, and $e_1, e_2, \cdots, e_n$ are the {\em stream 
expressions} over $s_1, \cdots, s_n, t_1, \cdots, t_m$. \qed
\end{definition}

Typically, {\em Input} streams are referred to as independent variables, 
whereas {\em output} streams are referred as dependent variable.
For example, consider the following \lola specification, where $t_1$ and 
$t_2$ are independent 
stream variables of type boolean and $t_3$ is an independent stream variable of type integer.
\begin{align*}
s_1 &= \tru \\
s_2 &= t_1 \lor (t_3 \leq 1) \\
s_3 &= \ite(s_3, s_4, s_4+1) \\
s_4 &= s_9[-1, 0] + (t_3 \mod 2) \\
\end{align*}
where, $\ite$ is the abbreviated form of {\it if-then-else} and stream 
expressions 
$s_7$ and $s_8$ refers to the stream $t_1$ with an offset of $+1$ and 
$-1$, respectively. Due to 
space constrains we present the full syntax of \lola in 
Appendix~\ref{sec:moreLOLA}.

The semantics of \lola specifications is defined in terms of the evaluation model, which 
describes the relation between input and output streams.

\begin{definition}\label{def:semantics-lola}
Given a \lola specification $\varphi$ over independent variables, $t_1, \cdots, t_m$, of type, 
$\type_1, \cdots, \type_m$, and dependent variables, $s_1, \cdots, s_n$ with type, 
$\type_{m+1}, \cdots, \type_{m+n}$, let $\tau_1, \cdots, \tau_m$ be the streams of length $N+1$, 
with $\tau_i$ of type $\type_i$. The tuple $\langle \stream_1, \cdots, \stream_n \rangle$ 
of streams of length $N+1$ is called the {\em evaluation model}, if for every equation in $\varphi$
$$s_i = e_i(t_1, \cdots, t_m, s_1, \cdots, s_n)$$
$\langle \stream_1, \cdots, \stream_n \rangle$ satisfies the following associated equations:
$$\stream_i(j) = \val(e_i)(j)~~~~~\text{ for } (1 \leq i \leq n) \wedge (0 \leq j \leq N)$$
where $\val(e_i)(j)$ is defined as follows. For the base cases:
\begin{align*}
\val(c)(j) &= c \\
\val(t_i)(j) &= \tau_i(j) \\
\val(s_i)(j) &= \stream_i(j)
\end{align*}
For the inductive cases, where $f$ is a function (e.g., arithmetic):
\begin{align*}
\val\Big(f(e_1, \cdots, e_k)\Big)(j) &= f\Big(\val(e_1)(j), \cdots, \val(e_k)(j)\Big) \\
\val\Big(\ite(b, e_1, e_2)\Big)(j) &= \mathsf{if }~ \val(b)(j) ~\mathsf{ then }~ \val(e_1)(j) \\
&~~~~~~\mathsf{ else }~ \val(e_2)(j) \\
\val(e[k,c])(j) &= 
\begin{cases}
\val(e)(j+k) & \textit{if } 0 \leq j + k \leq N \\
c  			 & \textit{otherwise} ~~~~~~~\blacksquare
\end{cases}
\end{align*} 
\end{definition}

The set of all equations associated with $\varphi$ is noted by $\varphi_\stream$.

\begin{definition}\label{def:lola-dgraph}
A {\em dependency graph} for a \lola specification, $\varphi$ is a weighted and directed graph
$G = \langle V, E \rangle$, with vertex set $V = \{s_1, \cdots, s_n, t_1, \cdots, t_m\}$. An 
edge $e : \langle s_i, s_k, w \rangle$ (resp. $e : \langle s_i, t_k, w \rangle$) labeled with 
a weight $w$ is in $E$ iff the equation for $\stream_i(j)$ in $\varphi_\stream$ contains 
$\stream_k(j+w)$ (resp. $\tau_k(j+w)$) as a subexpression. Intuitively, an edge records that 
$s_i$ at a particular position depends on the value of $s_k$ (resp. $t_k$), offset by $w$ positions.
\end{definition}

Given a set of synchronous input streams $\{ \stream_1, \stream_2, \cdots, \stream_m\}$ 
of respective type $\allType = \{ \type_1, \type_2, \cdots, \type_m\}$
and a \lola specification, $\varphi$, we evaluate the \lola specification, 
given by: 
$$(\stream_1, \stream_2, \cdots, \stream_m) \models_S \varphi$$
given the above semantics, where $\models_S$ denotes the synchronous 
evaluation.

%% file: problem.tex
\section{Partially Synchronous \lola}
\label{sec:pslola}

In this section, we extend the semantics of \lola to one that can 
accommodate reasoning about distributed systems.

\subsection{Distributed Streams}

Here, we refer to a global clock which will act as the ``real" timekeeper. It is to 
be noted that the presence of this global clock is just for theoretical reasons and 
it is not available to any of the individual streams.

We assume a {\em partially synchronous} system of $n$ streams, denoted by 
$\allStream = \{ \stream_1, \stream_2, \cdots, \stream_n \}$. For each 
stream $\stream_i$, where $i \in [1, |\allStream|]$, the local clock can be 
represented as a monotonically increasing function 
$c_i: \wholePlusSet \rightarrow \wholePlusSet$, where $c_i(\globalC)$ is the value of the 
local clock at global time $\globalC$. Since we are dealing with discrete-time systems, for 
simplicity and without loss of generality, we represent time with non-negative integers 
$\wholePlusSet$. For any two streams $\stream_i$ and $\stream_j$, where 
$i \neq j$, we assume:
$$
\forall \globalC \in \wholePlusSet. \mid c_i(\globalC) - c_j(\globalC) \mid < \epsilon,
$$
where $\epsilon > 0$ is the maximum clock skew. The value of $\epsilon$ is constant and 
is known (e.g., to a monitor). This assumption is met by the presence of an 
off-the-shelf clock 
synchronization algorithm, like NTP~\cite{ntp}, to ensure bounded clock skew among all streams. 
The local state of stream $\stream_i$ at time $\RTime$ is given by 
$\stream_i(\RTime)$, where $\RTime = c_i(\globalC)$, 
that is the local time of occurrence of the event at some global time $\globalC$.

\begin{definition}
A {\em distributed stream} consisting of $\allStream = \{\stream_1, 
\stream_2,  \ldots, \stream_n\}$ streams of length $N+1$ is represented by 
the pair $(\Events,  \hb)$, where
$\Events$ is a set of all local states (i.e., $\Events = \cup_{i \in [1,n], j\in[0, 
N]} \stream_i(j)$) partially ordered 
by Lamport's 
happened-before ($\hb$) relation~\cite{hb1978}, subject to the partial synchrony assumption:

\begin{itemize}
    \item For every stream $\stream_i$, $1 \leq i \leq |\allStream|$, all the 
    events happening on it are 
    totally ordered, that is, 
    $$\forall i,j,k  \in \wholePlusSet: (j < k) \rightarrow (\stream_i(j) \hb \stream_i(k))$$
    
    \item For any two streams $\stream_i$ and $\stream_j$ and two 
    corresponding events $\stream_i({k}), \stream_j({l}) \in \Events$, if $k + 
    \epsilon < l$ then, $\stream_i({k}) \hb \stream_j({l})$, where $\epsilon$ is 
    the maximum clock skew.

    \item For events, $e$, $f$, and $g$, if $e \hb f$ and $f \hb g$, then $e 
    \hb g$.\qed
\end{itemize}
\end{definition}

\begin{definition}
Given a distributed stream $(\Events, \hb)$, a subset of events $\cc \subseteq 
\Events$ is said to form a {\em consistent cut} if and only if when $\cc$ contains an event $e$, then 
it should also contain all such events that happened before $e$. Formally, 
$$
\forall e,f \in \Events. (e \in \cc) \land (f \hb e) \rightarrow f \in \cc.~\blacksquare
$$
\end{definition}

The frontier of a consistent cut $\cc$, denoted by $\front(\cc)$ is the set of all events that 
happened last in each stream in the cut. That is, $\front(\cc)$ is a set of 
$\stream_i(\textit{last})$ for each 
$i \in [1, |\allStream|]$ and $\stream_i(\textit{last}) \in \cc$. We denote 
$\stream_i(\textit{last})$ as the last event in $\stream_i$ 
such that $\forall \stream_i(\RTime) \in \cc. (\stream_i(\RTime) \neq \stream_i(\textit{last})) 
\rightarrow (\stream_i(\RTime) \hb \stream_i(\textit{last}))$.

\subsection{Partially Synchronous \lola}

We define the semantics of \lola specifications for partially synchronous 
distributed streams in terms of 
the evaluation model. The absence of a common global clock among the stream variables and the 
presence of the clock synchronization makes way for the output stream having multiple values at 
any given time instance. Thus, we update the evaluation model, so that $\alpha_i(j)$ and 
$\val(t_i)(j)$ are now defined by {\em sets} rather than just a single value.
This is due to nondeterminism caused by partial synchrony, i.e., the bounded clock skew $\epsilon$.

\begin{definition}\label{def:psync-eval-model}
Given a \lola~\cite{lola05} specification $\varphi$ over independent variables, 
$t_1, \cdots, t_m$ of type 
$\type_1, \cdots, \type_m$ and dependent variables, $s_1, \cdots, s_n$ of type 
$\type_{m+1}, \cdots, \type_{m+n}$ and $\tau_1, \cdots, \tau_m$ be the streams of length $N+1$, 
with $\tau_i$ of type $\type_i$. The tuple of streams $\langle\stream_1, \cdots, \stream_n\rangle$ of 
length $N+1$ with corresponding types is called the evaluation model in the partially synchronous 
setting, if for every equation in $\varphi$: 
$$s_i = e_i(t_1, \cdots, t_m, s_1, \cdots, s_n),$$
$\langle\stream_1, \cdots, \stream_n\rangle$ satisfies the following associated equations:
$$\stream_i(j) = \big\{\val(e_i)(k) \mid \max\{0,j-\epsilon+1\} \leq k \leq \min\{N,j+\epsilon-1\}\big\}$$
where $\val(e_i)(j)$ is defined as follows. For the base cases:
\begin{align*}
\val(c)(j) &= \{c\} \\
\val(t_i)(j) &= \big\{\tau_i(k) \mid \max\{0,j-\epsilon+1\} \leq k \leq \min\{N,j+\epsilon-1\}\big\} \\
\val(s_i)(j) &= \stream_i(j)
\end{align*}
For the inductive cases:
\begin{align*}
\val\Big(f(e_1, \cdots, e_p)\Big)(j) &= \Big\{ f(e_1', \cdots, e_p') \mid e_1' \in \val(e_1)(j), \cdots, \\
&~~~~~~~e_p' \in \val(e_p)(j) \Big\} \\
\val(\ite(b, e_1, e_2))(j) &= 
\begin{cases}
\val(e_1)(j) & \tru \in \val(b)(j) \\
\val(e_2)(j) & \fals \in \val(b)(j) \\
\end{cases} \\
\val(e[k,c])(j) &= 
\begin{cases}
\val(e)(j+k) & \textit{if } 0 \leq j + k \leq N \\
c            & \textit{otherwise}
\end{cases}
\end{align*} \qed
\end{definition}

\begin{example} 
Consider the \lola specification, $\varphi$, over the independent boolean variables \texttt{read} and 
\texttt{write}:
\begin{lstlisting}
input read:bool
input write:bool
output countRead := ite(read, countRead[-1,0] + 1, countRead[-1,0])
output countWrite := ite(write, countWrite[-1,0] + 1, countWrite[-1,0])
output check := (countWrite - countRead) <= 2
\end{lstlisting}
In Fig.~\ref{fig:psync-lola}, we have two input stream \textit{read} and \textit{write} which denotes the time instances where the corresponding events take place. It can be imagined that \textit{read} and \textit{write} are streams of type \texttt{boolean} with $\tru$ values at time instances $4, 6, 7$ and $2, 3, 5, 6$ and $\fals$ values at all other time instances respectively. We evaluate the above mentioned \lola specification considering a time synchronization constant, $\epsilon = 2$. The corresponding associated equations, $\varphi_\stream$, are:
\begin{align*}
\mathit{countRead}(j) &= 
\begin{cases}
\mathtt{ite}(\mathit{read}, 1, 0) & j = 0\\
\mathtt{ite}\Big(\mathit{read}, \mathit{countRead}(j- \\
~~~~~1) + 1, \mathit{countRead}(j)\Big) & j \in [1, N)
\end{cases} \\
\mathit{countWrite}(j) &= 
\begin{cases}
\mathtt{ite}(\mathit{write}, 1, 0) & j = 0\\
\mathtt{ite}\Big(\mathit{write}, \mathit{countWrite}(j- \\
~~~~~1) + 1, \mathit{countWrite}(j)\Big) & j \in [1, N)
\end{cases} \\
\mathit{check}(j) &= \Big(\mathit{countWrite}(j) - \mathit{countRead}(j)\Big) \leq 2
\end{align*}

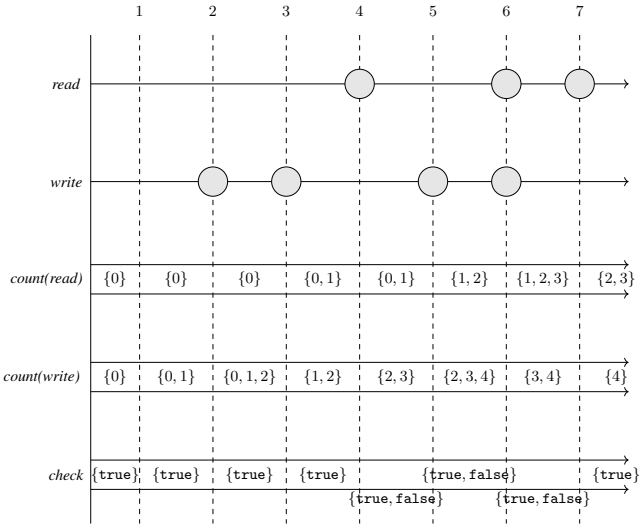
\begin{figure}
    \centering
    \input{fig_psync-lola}
    \caption{Partially Synchronous \lola Example}
    \label{fig:psync-lola}
    \vspace{-5mm}
\end{figure}

\end{example}

Similar to the synchronous case, evaluation of the partially synchronous \lola specification involves 
creating the dependency graph.

\begin{definition}\label{def:depend-graph}
A dependency graph for a \lola specification, $\varphi$ is a weighted directed multi-graph 
$G = \langle V,E \rangle$, with vertex set $V = \{s_1, \cdots, s_n, t_1, \cdots, t_m\}$. An edge 
$e : \langle s_i, s_k, w \rangle$ (resp. $e : \langle s_i, t_k, w \rangle$) labeled with a weight 
$w = \{\omega \mid p - \epsilon < \omega < p + \epsilon\}$ is in $E$ iff the equation for 
$\stream_i(j)$ contains $\stream_k(j + p)$ (resp. $\tau_k(j + p)$) as a sub-expression, for some $j$ 
and offset $p$.
\qed
\end{definition}

Intuitively, the dependency graph records that evaluation of a $s_i$ at a particular position 
depends on the value of $s_k$ (resp. $t_k$), with an offset in $w$. It is to be noted that there 
can be more than one edge between a pair of vertex $(s_i, s_k)$ (resp. $(s_i, t_k)$). Vertices 
labeled by $t_i$ do not have any outgoing edges.

\begin{example}
Consider the \lola specification over the independent integer variable \texttt{a}:

\begin{lstlisting}
input a : uint
output b1 := b2[1, 0] + ite(b2[-1,7] <= a[1, 0], b2[-2,0], 6)
output b2 := b1[-1,8]
\end{lstlisting}
Its dependency graph, shown in Fig.~\ref{fig:lola-dgraph} for $\epsilon = 
2$, has 1 edge from 
\texttt{b1} to \texttt{a} with a weight $\{0, 1, 2\}$. Similarly, there are 3 edges from 
\texttt{b1} to \texttt{b2} with 
weights $\{0, 1, 2\}, \{-2, -1, 0\}$ and $\{-3, -2, -1\}$ and 1 edge from \texttt{b2} 
to \texttt{b1} with a weight of $\{-2, -1, 0\}$

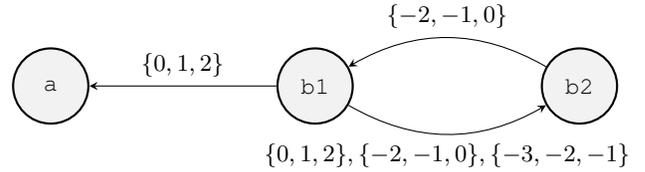
\begin{figure}
    \centering
    \input{fig_lola-dgraph}
    \caption{Dependency Graph Example}
    \label{fig:lola-dgraph}
    \vspace{-3mm}
\end{figure}

\end{example}

Given a set of partially synchronous input streams $\{ \stream_1, 
\stream_2, 
\cdots, \stream_{|\allStream|}\}$ of respective type $\allType = \{ \type_1, \type_2, \cdots, 
\type_{|\allStream|}\}$
 and a \lola specification, $\varphi$, the evaluation of $\varphi$ is given by 
$$(\stream_1, \stream_2, \cdots, \stream_{|\allStream|}) \models_{PS} \varphi$$
where, $\models_{PS}$ denotes the partially synchronous evaluation.

\section{Decentralized Monitoring Architecture}
\label{sec:problem}

\subsection{Overall Picture}
\label{subsec:overallProb}

We consider a decentralized online monitoring system comprising of a fixed number of $|\Monitors|$ reliable monitor 
processes $\Monitors = \{M_1, M_2, \cdots, M_{|\Monitors|}\}$ that can communicate with each other by sending and 
receiving messages through a complete point-to-point bidirectional communication links. Each communication 
link is also assumed to be reliable, i.e., there is no loss or alteration of messages. Similar to the 
distributed system under observation, we assume the clock on the individual monitors are  
asynchronous, with clock synchronization constant = $\epsilon_M$.

Throughout this section we assume that the global distributed stream 
consisting of complete observations of $|\allStream|$ streams is only partially visible to each 
monitor. Each monitor process locally executes an identical sequential algorithm which consists of the 
following steps (we will generalize this approach in Section~\ref{sec:mon}). In other words, an 
evaluation iteration of each monitor consists of the following steps:

\begin{enumerate}
    \item Reads the a subset of $\Events$ events (visible to $M_i$) along with 
    the corresponding time and valuation of the events, which results in the construction 
    of a {\em partial distributed stream};

    \item Each monitor evaluates the \lola specification $\varphi$ given the partial distributed 
    stream;
    
    \item Every monitor, broadcasts a message containing rewritten associated equations of 
    $\varphi$, denoted $\LS$, and
    
    \item Based on the message received containing associated equations, each monitor 
    amalgamates the observations of all the monitors to compose a set of associated equations. 
    After a 
    evaluation iteration, each monitor will have the same set of associated equations to be evaluated 
    on the upcoming distributed stream.
\end{enumerate}

The message sent from monitor $M_i$ at time $\LTime$ to another monitor $M_j$, for all 
$i, j \in [1, |\Monitors|]$, during a evaluation iteration of the monitor is assumed to reach latest by time 
$\LTime + \epsilon_M$. Thus, the length of an {\em evaluation iteration} $k$ can be adjusted to make 
sure the message from all other monitors reach before the start of the next evaluation iteration.

\subsection{Detailed Description}
\label{subsec:detprob}

We now explain in detail the computation model (see Algorithm~\ref{alg:monitor}). Each monitor 
process $M_i \in \Monitors$, where 
$i \in [1, |\Monitors|]$, attempts to read $e \in \Events$, given the distributed stream, 
$(\Events, \hb)$. An event can either be observable, 
or not observable. Due to distribution, this results in obtaining a partial 
distributed stream $(\Events_i, \hb)$ defined below.

\begin{algorithm}[t]
	\caption{Behavior of a Monitor $M_i$, for $i \in [1, |\Monitors|]$}\label{alg:monitor}
	\input{alg_monitor}
\end{algorithm}

\begin{definition}
Let $(\Events, \hb)$ be a distributed stream. We say that $(\Events', \hb)$ is a {\em partial 
distributed stream} for $(\Events, \hb)$ and denote it by $(\Events', 
\hb) \sqsubseteq (\Events, \hb)$ iff $\Events' \subseteq \Events$ (the happened before relation is 
obviously preserved). 
%
%
\qed
\end{definition}

We now tie partial distributed streams to a set of decentralized monitors and the fact that 
decentralized monitors can only partially observe a distributed stream. First, all un-observed 
events is replaced by $\natural$, i.e., for all $\stream_i(\RTime) \in \Events$ if 
$\stream_i(\RTime) \not\in \Events_i$ then $\Events_i = \Events_i \cup \{\stream_i(\RTime) = \natural\}$.

\begin{definition}
\label{def:view_const}
Let $(\Events, \hb)$ be a distributed stream and $\Monitors = \{M_1, M_2, \cdots, 
M_{|\Monitors|}\}$ be a set of monitors, where each monitor $M_i$, for $i \in [1, |\Monitors|]$ is 
associated with a partial distributed stream $(\Events_i, \hb) \sqsubseteq (\Events, \hb)$.
We say that these monitor observations are consistent if 
\begin{itemize}
    \item $\forall e \in \Events. \exists i \in [1, |\Monitors|]. e \in \Events_i$, and
    \item $\forall e \in \Events_i. \forall e' \in \Events_j. (e = e' \land e \neq \natural) \oplus \Big( (e = 
    \natural \lor e' = \natural) \Big)$,
\end{itemize}
where $\oplus$ denoted the exclusive-or operator.
\end{definition}

%
In a partially synchronous system, there are different ordering of events and each unique ordering of events might evaluate to different values. Given a distributed stream, $(\Events, \hb)$, a sequence of consistent cuts is of the form $\cc_0\cc_1\cc_2\cdots\cc_N$, where for all $i \geq 0$:
(1) $\cc_i \subseteq \Events$, and 
(2) $\cc_i \subseteq \cc_{i+1}$. 

Given the semantics of partially-synchronous \lola, evaluation of output stream variable $s_i$ at 
time instance $j$ requires events $\stream_i(k)$, where $i \in [1, |\allStream|]$ and 
$k \in \Big\{\LTime \mid \max\{0, j - \epsilon + 1 \} \leq \LTime \leq \{N, j + \epsilon - 1\}\Big\}$. To translate monitoring of a distributed stream to a synchronous stream, we make sure that the 
events in the frontier of a consistent cut, $\cc_j$ are $\stream_i(k)$.

Let $\ccAll$ denote the set of all valid sequences of consistent cuts. We define the set of all 
synchronous streams of $(\Events, \hb)$ as follows:
$$\Sr(\Events, \hb) = \Big\{\front(\cc_0)\front(\cc_1)\cdots \mid \cc_0\cc_1\cdots \in \ccAll \Big\}$$
Intuitively, $\Sr(\Events, \hb)$ can be interpreted as the set of all possible 
``interleavings''.
%
The evaluation of the \lola  specification, $\varphi$, with respect to $(\Events, \hb)$ is the following :
\begin{align*}
\Big[(\Events, \hb) \models_{PS} \varphi \Big] &= \Big\{(\stream_1, \cdots, \stream_n) \models_{S} 
\varphi \mid (\stream_1, \cdots, \stream_n) \in \\
&~~~~~~~~~~\Sr(\Events, \hb) \Big\}
\end{align*}
This means that evaluating a partially synchronous distributed stream with respect to a \lola specification results 
in a set of evaluated results, as the computation may involve several streams. This also enables reducing the problem from evaluation of a partially synchronous distributed system to the evaluation of multiple synchronous streams, each evaluating to unique values for the output stream, with message complexity
$$O\big(\epsilon^{|\allStream|} N |\Monitors|^2\big) \;\;\; \Omega(N |\Monitors|^2)$$

\subsection{Problem Statement}
\label{subsec:prob}

The overall problem statement requires that upon the termination of 
the Algorithm~\ref{alg:monitor}, the verdict of all the monitors in the decentralized monitoring architecture is the 
same as that of a centralized monitor which has the global view of the system
$$\forall i \in [1, m]: \Result_i = \Big[(\Events, \hb) \models_{PS} \varphi\Big]$$
where $(\Events, \hb)$ is the global distributed stream and $\varphi$ is the \lola 
specification with $\Result_i$ as the evaluated result by monitor $M_i$.

%% file: fig_psync-lola.tex
    \centering 
    \scalebox{0.65}{
        \begin{tikzpicture}
        
            \tikzstyle{every node}=[font=\small]

            \draw [-] (0,10) -- (0, 0);
            
            \draw (-0.5,9) node[] {\textit{read}};
            \draw (-0.5,7) node[] {\textit{write}};
            \draw (-0.9,5) node[] {\textit{count(read)}};
            \draw (-1,3) node[] {\textit{count(write)}};
            \draw (-0.5,1) node[] {\textit{check}};
            
            \draw [->] (0,9) -- (11,9);
            \draw [->] (0,7) -- (11,7);

            \draw[dashed] (1, 10) -- (1, 0);
            \draw (1, 10.5) node[] {$1$};
            \draw[dashed] (2.5, 10) -- (2.5, 0);
            \draw (2.5, 10.5) node[] {$2$};
            \draw[dashed] (4, 10) -- (4, 0);
            \draw (4, 10.5) node[] {$3$};
            \draw[dashed] (5.5, 10) -- (5.5, 0);
            \draw (5.5, 10.5) node[] {$4$};
            \draw[dashed] (7, 10) -- (7, 0);
            \draw (7, 10.5) node[] {$5$};
            \draw[dashed] (8.5, 10) -- (8.5, 0);
            \draw (8.5, 10.5) node[] {$6$};
            \draw[dashed] (10, 10) -- (10, 0);
            \draw (10, 10.5) node[] {$7$};
            
            \draw[fill=black!10] (5.5,9) circle (0.3) node[above, yshift=0.1cm]{};
            \draw[fill=black!10] (8.5,9) circle (0.3) node[above, yshift=0.1cm]{};
            \draw[fill=black!10] (10,9) circle (0.3) node[above, yshift=0.1cm]{};

            \draw[fill=black!10] (2.5,7) circle (0.3) node[above, yshift=0.1cm]{};
            \draw[fill=black!10] (4,7) circle (0.3) node[above, yshift=0.1cm]{};
            \draw[fill=black!10] (7,7) circle (0.3) node[above, yshift=0.1cm]{};
            \draw[fill=black!10] (8.5,7) circle (0.3) node[above, yshift=0.1cm]{};

            \draw [->] (0, 5.3) -- (11, 5.3);
            \draw [->] (0, 4.7) -- (11, 4.7);
            \draw (0.5, 5) node[] {$\{0\}$};
            \draw (1.75, 5) node[] {$\{0\}$};
            \draw (3.25, 5) node[] {$\{0\}$};
            \draw (4.75, 5) node[] {$\{0, 1\}$};
            \draw (6.25, 5) node[] {$\{0, 1\}$};
            \draw (7.75, 5) node[] {$\{1, 2\}$};
            \draw (9.25, 5) node[] {$\{1, 2, 3\}$};
            \draw (10.75, 5) node[] {$\{2, 3\}$};

            \draw [->] (0, 3.3) -- (11, 3.3);
            \draw [->] (0, 2.7) -- (11, 2.7);
            \draw (0.5, 3) node[] {$\{0\}$};
            \draw (1.75, 3) node[] {$\{0, 1\}$};
            \draw (3.25, 3) node[] {$\{0, 1, 2\}$};
            \draw (4.75, 3) node[] {$\{1, 2\}$};
            \draw (6.25, 3) node[] {$\{2, 3\}$};
            \draw (7.75, 3) node[] {$\{2, 3, 4\}$};
            \draw (9.25, 3) node[] {$\{3, 4\}$};
            \draw (10.75, 3) node[] {$\{4\}$};

            \draw [->] (0, 1.3) -- (11, 1.3);
            \draw [->] (0, 0.7) -- (11, 0.7);
            \draw (0.5, 1) node[] {$\{\tru\}$};
            \draw (1.75, 1) node[] {$\{\tru\}$};
            \draw (3.25, 1) node[] {$\{\tru\}$};
            \draw (4.75, 1) node[] {$\{\tru\}$};
            \draw (6.25, 1) node[yshift=-5mm] {$\{\tru, \fals\}$};
            \draw (7.75, 1) node[] {$\{\tru, \fals\}$};
            \draw (9.25, 1) node[yshift=-5mm] {$\{\tru, \fals\}$};
            \draw (10.75, 1) node[] {$\{\tru\}$};
        \end{tikzpicture}
    }

%% file: fig_lola-dgraph.tex
    \tikzset{
        ->,
        >=stealth,
        node distance=100,
        every state/.style={thick,
        fill=gray!10},
        initial text=$ $
    }
    \tikzstyle{every node}=[font=\small]
    \centering
    \begin{tikzpicture}
        
        \node[state, minimum size=1cm] (1) {\texttt{a}};
        \node[state, minimum size=1cm, right of=1] (2) {\texttt{b1}};
        \node[state, minimum size=1cm, right of=2] (3) {\texttt{b2}};
        
        \draw (2) edge node[above]{$\{0, 1, 2 \}$} (1);
        \draw (2) edge[bend right] node[below]{$\{0, 1, 2\}, \{-2, -1, 0\}, \{-3, -2, -1\}$} (3);
        \draw (3) edge[bend right] node[above]{$\{-2, -1, 0 \}$} (2);
        
    \end{tikzpicture}

%% file: alg_monitor.tex
	\begin{algorithmic}[1]
		\footnotesize
		\For{$j = 0$ to $N$}
			\State Let $(\Events_i, \hb_i)_j$ be the partial distributed stream view of $M_i$
			\State $\LS_j \leftarrow \big[(\Events, \hb) \models_{PS} \varphi_\stream \big]$
			\State {\bf Send: } broadcasts symbolic view $\LS_j$
			\State {\bf Receive: } $\Pi_j \leftarrow \{\LS^k_j \mid 1 \leq k \leq \Monitors\}$
			\State {\bf Compute: } $\LS_{j+1} \leftarrow \LC(\Pi_j)$
		\EndFor
	\end{algorithmic}

%% file: progression.tex
\section{Calculating $\LS$}
\label{sec:progress}


In this section, we introduce the rules of rewriting \lola associated equations given the 
evaluated results and observations of the system. In our distributed 
setting, evaluation of a \lola specification involves generating a set of synchronous streams and evaluating 
the given \lola specification on it (explained in Section~\ref{sec:smt}). Here, we make use of  
the evaluation of \lola specification into forming our local observation to be shared with other 
monitors in the system.


Given the set of synchronous streams, $(\stream_1, \stream_2, \cdots, \stream_{|\allStream|})$, 
the symbolic locally computed result $\LS$ (see Algorithm~\ref{alg:monitor}) consists of associated \lola equations, which either needs more 
information (data was unobserved) from other monitors to evaluate or the concerned monitor needs 
to wait (positive offset). In either case, the associated \lola specification is shared with all 
other monitors in the system as the missing data can be observed by either monitors.
We divide the rewriting rules into three cases, depending upon the observability of the value 
of the independent variables required for evaluating the expression $e_i$ for all 
$i \in [1, n]$. Each stream expression is categorized into three cases (1) completely unobserved, 
(2) completely observed or (3) partially observed. This can be done easily by going over the 
dependency graph and checking with the partial distributed stream read by the corresponding 
monitor.


\noindent{\bf Case 1 (Completely Observed). } Formally, a completely observed stream 
expression $s_i$ can be identified from the dependency graph, $G = \langle V,E \rangle$, as for all 
$s_k$ (resp. $t_k$) 
$\langle s_i, s_k, w \rangle \in E$ (resp. $\langle s_i, t_k, w \rangle \in E$), 
$s_k(j+w) \neq \natural$ (resp. $t_k(j+w) \neq \natural$) are observed for time instance $j$. If yes, this signifies, that 
all independent and dependent variables required to evaluate $s_i(j)$, is observed by the monitor $M$, there by evaluating:
$s_i(j) = e_i(s_1, \cdots, s_n, t_1, \cdots, t_m)$
and rewriting $s_i(j)$ to $\LS$. 

\noindent{\bf Case 2 (Completely Unobserved). } Formally, we present a completely unobserved stream 
expression, $s_i$ from the dependency graph, $G = \langle V,E \rangle$, as for all 
$s_k$ (resp. $t_k$), 
$\langle s_i, s_k, w \rangle \in E$ (resp. $\langle s_i, t_k, w \rangle \in E$), 
$s_k(j+w) = \natural$ (resp. $t_k(j+w) = \natural$) are unobserved, for time instance $j$ . This signifies that the valuation of 
neither variables are known to the monitor $M$. Thus, we rewrite the following stream expressions
\begin{align*}
s_k'(j) &=
\begin{cases}
s_k(j+w) & 0 \leq j+w \leq N \\
\mathtt{default} & \text{ otherwise}
\end{cases} \\
t_k'(j) &=
\begin{cases}
t_k(j+w) & 0 \leq j+w \leq N \\
\mathtt{default} & \text{ otherwise}
\end{cases}
\end{align*}
for all $\langle s_i, s_k, w \rangle \in E$ and $\langle s_i, t_k, w \rangle \in E$, and include the rewritten associated equation for evaluating $s_i(j)$ as
$$s_i(j) = e_i(s_1', \cdots, s_n', t_1', \cdots, t_m')$$
It is to be noted that the $\mathtt{default}$ value of a stream variable, $s_k$ (resp. $t_k$), 
depends on the corresponding type $\type_k$ (resp. $\type_{m+k}$) of the stream.

\noindent{\bf Case 3 (Partially Observed). } Formally, we present a partially observed stream 
expression, $s_i$ from the dependency graph, $G = \langle V,E \rangle$, as for all 
$s_k$ (resp. $t_k$), they are either observed or unobserved, for time instance $j$. In other words, we can represent 
a set $\mathbb{V}_o = \{s_k \mid \exists s_k(j+w) \neq \natural \}$ 
of all observed stream variable and a set 
$\mathbb{V}_u = \{s_k \mid s_k(j+w) = \natural \}$ of 
all unobserved dependent stream variable for all $\langle s_i, s_k, w \rangle \in E$. The set 
can be expanded to include independent variables as well.
For all $s_k \in \mathbb{V}_u$ (resp. $t_k \in \mathbb{V}_u$) that are unobserved, are replaced by:
\begin{align*}
s_k^u(j) &=
\begin{cases}
s_k(j+w) & 0 \leq j+w \leq N \\
\mathtt{default} & \text{ otherwise}
\end{cases} \\
t_k^u(j) &=
\begin{cases}
t_k(j+w) & 0 \leq j+w \leq N \\
\mathtt{default} & \text{ otherwise}
\end{cases}
\end{align*}
and for all $s_k \in \mathbb{V}_o$ (resp. $t_k \in \mathbb{V}_o$) that are observed, 
are replaced by:
\begin{align*}
s_k^o(j+w) &= \mathtt{value} \\
t_k^o(j+w) &= \mathtt{value}
\end{align*}
and there by partially evaluating $s_i(j)$ as
$$s_i(j) = e_i(s^o_1, \cdots, s^o_n, t^o_1, \cdots, t^o_m, s^u_1, \cdots, s^u_n, t^u_1, \cdots, t^u_m)$$
followed by adding the partially evaluated associated equation for $s_i(j)$ to $\LS$. It is to be noted, that a
consistent partial distributed stream makes sure that for all $s_k$ (resp. $t_k$), 
can only be either observed or unobserved and not both or neither.

\begin{example}
\label{example:prog}
Consider the \lola specification mentioned below and the stream input of length $N = 6$ 
divided into two evaluation rounds and $\epsilon = 2$ as shown in 
Fig.~\ref{fig:progExample} with the monitors $M_1$ and $M_2$. 

\begin{lstlisting}
input a : uint
input b : uint
output c := ite(a[-1,0] <= b[1, 0], a[1,0], b[-1, 0])
\end{lstlisting}

The associated equation for the output stream is:
\begin{equation}
\nonumber
c = 
\begin{cases}
\mathtt{ite}(0 \leq b(i+1), a(i+1), 0) & i = 1 \\
\mathtt{ite}(a(i-1) \leq b(i+1), a(i+1), \\
~~~~~~~~~~b(i-1)) & 2 \leq i \leq N-1 \\
\mathtt{ite}(a(i-1) \leq 0, 0, b(i-1)) & i = N
\end{cases}
\end{equation}
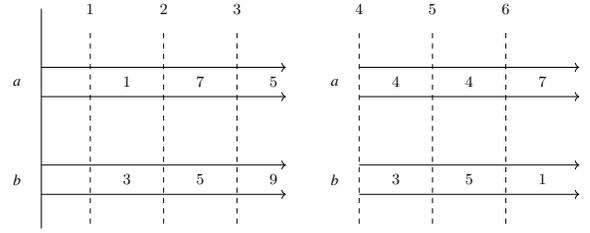
\begin{figure}
	\input{fig_progExample}
	\caption{Example of generating $\LS$}
	\label{fig:progExample}
	\vspace{-5mm}
\end{figure}
Let the partial distributed stream read by monitor $M_1$ include $\{a, (1, 1), (3, 5)\}, \{b, (2, 5), (3, 9)\}$ and the partial distributed stream read by monitor $M_2$ include $\{a, (1, 1), (2, 7)\}, \{b, (1, 3), (3, 9)$. Monitor $M_1$ evaluates $c(2) = 5$ and partially evaluates $c(1)$ and $c(3)$. Thus $\LS^1_1 = \{c(1) = a(2), c(2) = 5, c(3) = \ite(a(2) \leq b(4), a(4), 5) \}$. Monitor $M_2$ partially evaluates all $c(1)$, $c(2)$ and $c(3)$ and thus $\LS^2_1 = \{c(1) = \ite(0 \leq b(2), a(2), 0), c(2) = a(3), c(3) = \ite(7 \leq b(4), a(4), b(2))\}$.

Let the partial distributed stream read by monitor $M_1$ include $\{a, (4, 4), (5, 4)\}, \{b, (4, 3), (6, 1)\}$ and the partial distributed stream read by monitor $M_2$ include $\{a, (5, 4), (6, 7)\}, \{b, (4, 3), (5, 5)\}$. Monitor $M_1$ evaluates $c(4) = 9$ and $c(5) = 3$ and partially evaluates $c(6)$. Thus $\LS^1_2 = \{c(4) = 9, c(5) = 3, c(6) = b(5) \}$. Monitor $M_2$ evaluates $c(6) = 5$ and partially evalues $c(4)$ and $c(5)$ and thus $\LS^2_2 = \{c(4) = \ite(a(3) \leq 5, 4, 9), c(5) = \ite(a(4) \leq b(6), 7, 3), c(6) = 5 \}$.

It is to be noted, the after the first round of evaluation, the corresponding local states, $\LS^1_1$ and $\LS^2_1$ will be shared which will enable evaluating the output stream for few of the partially evaluated output stream (will be discussed in Section~\ref{subsec:union}). These will be included in the local state of the following evaluation round.

\end{example}

Note that generating $\LS$ takes into consideration an ordered stream. One where 
the time of occurrence of events and values are comparable. It can be imagined that generating 
the same for 
the distributed system involves generating it for all possible ordering of events. This 
will be discussed in details in the following sections.s.

%% file: fig_progExample.tex
    \centering 
    \scalebox{0.65}{
        \begin{tikzpicture}
        
            \tikzstyle{every node}=[font=\small]

            \draw [-] (0,4.5) -- (0, 0);
            
            \draw (-0.5,3) node[] {\textit{a}};
            \draw (-0.5,1) node[] {\textit{b}};
            
            \draw [->] (0,3.3) -- (5,3.3);
            \draw [->] (0,2.7) -- (5,2.7);
            \draw [->] (0,1.3) -- (5,1.3);
            \draw [->] (0,0.7) -- (5,0.7);

            \draw (6,3) node[] {\textit{a}};
            \draw (6,1) node[] {\textit{b}};
            
            \draw [->] (6.5,3.3) -- (11,3.3);
            \draw [->] (6.5,2.7) -- (11,2.7);
            \draw [->] (6.5,1.3) -- (11,1.3);
            \draw [->] (6.5,0.7) -- (11,0.7);

            \draw[dashed] (1, 4) -- (1, 0);
            \draw (1, 4.5) node[] {$1$};
            \draw[dashed] (2.5, 4) -- (2.5, 0);
            \draw (2.5, 4.5) node[] {$2$};
            \draw[dashed] (4, 4) -- (4, 0);
            \draw (4, 4.5) node[] {$3$};

            \draw[dashed] (6.5, 4) -- (6.5, 0);
            \draw (6.5, 4.5) node[] {$4$};
            \draw[dashed] (8, 4) -- (8, 0);
            \draw (8, 4.5) node[] {$5$};
            \draw[dashed] (9.5, 4) -- (9.5, 0);
            \draw (9.5, 4.5) node[] {$6$};

            \draw (1.75, 3) node[] {$1$};
            \draw (3.25, 3) node[] {$7$};
            \draw (4.75, 3) node[] {$5$};

            \draw (7.25, 3) node[] {$4$};
            \draw (8.75, 3) node[] {$4$};
            \draw (10.25, 3) node[] {$7$};

            \draw (1.75, 1) node[] {$3$};
            \draw (3.25, 1) node[] {$5$};
            \draw (4.75, 1) node[] {$9$};

            \draw (7.25, 1) node[] {$3$};
            \draw (8.75, 1) node[] {$5$};
            \draw (10.25, 1) node[] {$1$};

        \end{tikzpicture}
    }

%% file: smt.tex
\section{SMT-based Solution}
\label{sec:smt}

\subsection{SMT Entities}

SMT entities represent (1) \lola equations, and (2) variables used to 
represent the distributed stream. 
Once we have generated a sequence of consistent cuts, we use the laws 
discussed in Section~\ref{sec:progress}, to construct the set of all locally computer or 
partially computed \lola equations. 

\vspace{2mm}
\noindent \textbf{Distributed Stream.}
In our SMT encoding, the set of events, $\Events$, is represented by a bit vector, where each bit 
corresponds to an individual event in the distributed stream, $(\Events, \hb)$. The length 
of the stream under observation is $k$, which makes $|\Events| = k \times |\allStream|$ and 
the length of the entire stream is $N$. We 
conduct a pre-processing of the distributed stream where we create a $\Events \times \Events$ 
matrix, \texttt{hbSet} to incorporate the happen-before relations. We populate \texttt{hbSet} 
as \texttt{hbSet}[e][f] = 1 iff $e \hb f$, else \texttt{hbSet}[e][f] = 0. In order to map each 
event to its respective stream, we introduce a function, $\mu: \Events \rightarrow \allStream$.{\tiny {\tiny {\tiny {\tiny {\tiny {\tiny }}}}}}

We introduce a valuation function, $\val: \Events \rightarrow \type$ (whatever the type is in the \lola specification), in order to represent 
the values of the individual events. Due to the partially synchronous 
assumption of the system, the possible time of occurrence of an event is defined by a function 
$\delta: \Events \rightarrow \wholePlusSet$, where
$\forall \stream(\RTime) \in \Events. \exists \RTime' \in [\max\{0, \RTime - \epsilon + 1\}, \min\{\RTime + \epsilon - 1\}, N]. \delta\big(\stream(\RTime)\big) = \RTime'$.
We update the $\delta$ function when referring to events on output streams by updating the time 
synchronization constant to $\epsilon_M$. This accounts for the clock skew between two monitors. 
Finally, we introduce an 
uninterpreted function $\rho: \wholePlusSet \rightarrow 2^\Events$ that identifies a sequence of 
consistent cuts for computing all possible evaluations of the \lola 
specification, while satisfying a number of given constrains explained in 
Section~\ref{subsec:smt_const}.

\subsection{SMT Constrains}
\label{subsec:smt_const}

Once we have defined the necessary SMT entities, we move onto the SMT constraints. We first define the 
SMT constraints for generating a sequence of consistent cuts, followed by the ones for evaluating 
the given \lola equations $\varphi_\stream$.

\textbf{Constrains for consistent cuts over $\boldsymbol{\rho}$:} In order to make sure that the uninterpreted function $\rho$ identifies a sequence of consistent cuts, we enforce certain constraints. The first constraint enforces that each element in the range of $\rho$ is in fact a consistent cut:
$$\forall i \in [0, k]. \forall e, e' \in \Events. \Big( (e \hb e') \land (e' \in \rho(i)) \Big) \rightarrow (e \in \rho(i))$$
Next, we enforce that each successive consistent cut consists of all events included in the previous consistent cut:
$$\forall i \in [0, k-1]. \rho(i) \subseteq \rho(i+1)$$
Next, we make sure that the front of each consistent cut constitutes of events with possible time of occurrence in accordance with the semantics of partially-synchronous \lola:
$$\forall i \in [0, k]. \forall e \in \front(\rho(i)). \delta(e) = i$$
%
%
Finally, we make sure that every consistent cut consists of events from all streams:
$$\forall i \in [0, k]. \forall \stream \in \allStream. \exists e \in \front(\rho(i)). \mu(e) = \stream$$
\textbf{Constrains for \lola specification:} These constraints will evaluate the \lola 
specifications and will make sure that $\rho$ will not only represent a valid sequence of 
consistent cuts but also make sure that the sequence of consistent cuts evaluate the \lola equations, 
given the stream expressions. As is evident that a distributed system can often evaluate to multiple 
values at each instance of time. Thus, we would need to check for both satisfaction and violation 
for logical expressions and evaluate all possible values for arithmetic expressions. Note that 
monitoring all \lola specification can be reduce to evaluating expressions that are either logical 
or arithmetic. Below, we mention the SMT constraint for evaluating different \lola equations at time 
instance $j$:
\begin{align*}
t_i[p, c] &=
\begin{cases}
 \val(e) & 0 \leq j + p \leq N \\
 c & \text{ otherwise}
\end{cases} \\
&~~~~~~\Big(\exists e \in \front(\rho(j+p)). (\mu(e) = \stream_i)\Big) \\
s_i(j) &= \tru ~~ \front(\rho(j)) \models \varphi_\stream \\ 
&~~~~~~\text{(Logical expression, satisfaction)} \\
s_i(j) &= e_i(\forall e \in \front(\rho(j)). \val(e)) \\ 
&~~~~~~\text{(Arithmetic expression, evaluation)} \\
\end{align*}
The previously evaluated result is included in the SMT instance as a entity and a additional constrain is added that only evaluates to unique value, in order to generate all possible evaluations. The SMT instance returns a satisfiable result iff there exists at-least one unique evaluation of the equation. This is repeated multiple times until we are unable to generate a sequence of consistent cut, given the constraints, i.e., generate unique values. It is to be noted that stream expression of the form \texttt{ite}$(s_i, s_k, s_j)$ can be reduced to a set of expressions where we first evaluate $s_i$ as a logical expression followed by evaluating $s_j$ and $s_k$ accordingly.

%% file: monitor.tex
\section{Runtime Verification of \lola specifications}
\label{sec:mon}

Now that both the rules of generating rewritten \lola equations (Section~\ref{sec:progress}) and 
the working of the 
SMT encoding (Section~\ref{sec:smt}) have been discussed, we can finally bring them together 
in order to solve the problem introduced in Section~\ref{sec:problem}.

\subsection{Computing $\LC$}
\label{subsec:union}

Given a set of local states computed from the SMT encoding, each monitor process receives a set of rewritten \lola associated equations, denoted by $\LS^i_j$, where $i \in [1, |\Monitors|]$ for $j$-th computation round. Our idea to compute $\LC$ from these sets is to simply take a prioritized union of all the associated equations.
$$\LC(\Pi^i_j) = \biguplus_{i \in [1, |\Monitors|]} \LS^i_j$$
The intuition behind the priority is that an evaluated \lola equation will take precedence over a 
partially evaluated/unevaluated \lola equation, and two partially-evaluated \lola equation 
will be combined to form a evaluated or partially evaluated \lola equation. For example, taking 
the locally computed $\LS^1_1$ 
and $\LS^2_1$ from Example~\ref{example:prog}, $\LC(\LS^1_1, \LS^2_1)$ is computed to be 
$\{c(1) = a(2), c(2) = 5, c(3) = \ite(7 \leq b(4), a(4), 5)\}$ at Monitor $M_1$ and 
$\{c(1) = 7, c(2) = 5, c(3) = \ite(7 \leq b(4), a(4), 5)\}$ at Monitor $M_2$. Subsequently, 
$\LC(\LS^1_2, \LS^2_2)$ is computed to be $\{c(4) = 9, c(5) = 3, c(6) = 5\}$ at Monitor $M_1$ 
and $\{c(4) = 9, c(5) = 3, c(6) = 5\}$ at Monitor $M_2$.

\subsection{Bringing it all Together}
\label{subsec:final}

As stated in Section~\ref{subsec:overallProb}, the monitors are decentralized and online. Since, 
setting up of a SMT instance is costly (as seen in our evaluated results in 
Section~\ref{sec:eval}), we often find it more efficient to evaluate the \lola specification after 
every $k$ time instance. This reduces the number of computation rounds to $\lceil N/k \rceil$ as 
well as the number of messages being transmitted over the network as well with an increase to the 
size of the messages. We update Algorithm~\ref{alg:monitor} to reflect our solution more closely 
to Algorithm~\ref{alg:monitorUpdt}.

\begin{algorithm}[t]
	\caption{Computation on Monitor $M_i$}\label{alg:monitorUpdt}
\input{alg_monitor_updt}
\end{algorithm}

Each evaluation round starts by reading the $r$-th partial distributed system which consists of 
events occurring between the time $\max\{0, (r-1)\times\lceil N/k \rceil\}$ and 
$\min\{N, r\times\lceil N/k \rceil\}$ (line 3). We assume that the partial distributed system is 
consistent in accordance with the assumption that each event has been read by atleast one 
monitor. To account for any concurrency among the events in $(r-1)$-th computation 
round with that in the $r$-th computation round, we expand the length by  
$\epsilon$ time, there-by making the length of the $r$-th computation round, 
$\max\{0, (r-1)\times\lceil N/k \rceil - \epsilon + 1\}$ and $\min\{N, r\times\lceil N/k \rceil\}$.

Next, we reduce the evaluation of the distributed stream problem into an SMT problem (line 7). 
We represent the distributed system using SMT entities and then by the help of SMT constraints, and we 
evaluate the \lola specification on the generated sequence of consistent cuts. Each sequence of 
consistent cut presents a unique ordering of the events which evaluates to a unique value for the 
stream expression (line 8). This is repeated until we no longer can generate a sequence of 
consistent cut that evaluates $\varphi_\stream$ to unique values (line 9). Both the evaluated 
as well as partially evaluated results are included in $\LS$ as associated 
\lola equations. This is followed by the communication phase where each monitor shares 
its locally computed $\LS^i_r$, for all $i \in [1, |\Monitors|]$ and $r$ evaluation round (line 10-11).

Once, the local states of all the monitors are received, we take a prioritized union of all the associated equation and include them into $\LS^i_{r+1}$ set of associated equations (line 12). Following this, the computation shifts to next computation round and the above mentioned steps repeat again. Once we reach the end of the computation, all the evaluated values are contained in $\Result^i$

\begin{lemma}\label{thm:terminate}
Let $\allStream = \{S_1, S_2, \cdots, S_n \}$ be a distributed system and $\varphi$ be an \lola 
specification. Algorithm~\ref{alg:monitor} terminates when monitoring a terminating distributed 
system.
\end{lemma}


\begin{theorem}\label{thm:corr-sound}
Algorithm~\ref{alg:monitorUpdt} solves the problem stated in Section~\ref{sec:problem}.
\end{theorem}


\begin{theorem}\label{thm:msgComplex}
Let $\varphi$ be a \lola specification and $(\Events, \hb)$ be a distributed 
stream consisting of $|\allStream|$ streams. The message complexity of Algorithm~\ref{alg:monitorUpdt} with $|\Monitors|$ monitors is
$$O\big(\epsilon^{|\allStream|} N |\Monitors|^2\big) \;\;\; \Omega(N |\Monitors|^2)$$
\end{theorem}


%% file: alg_monitor_updt.tex
	\begin{algorithmic}[1]
		\footnotesize
		\State $\LS^i_1[0] = \emptyset$
		\For{$r = 1$ to $\lceil N/k \rceil$}
			\State $(\Events_i, \hb_i)_r \leftarrow$ $r$-th Consistent partial distributed stream
			\State $j = 0$
			\Do
			\State $j = j+1$
				\State $(\stream_1, \stream_2, \cdots, \stream_{|\allStream|}) \in \Sr(\Events_i, \hb_i)$
				\State $\LS^i_r[j] \leftarrow \LS^i_{r}[j-1] \cup \big[ (\stream_1, \stream_2, \cdots, \stream_{|\allStream|}) \models_S \varphi_\stream \big]$
			\doWhile{$(\LS^i_r[j] \neq \LS^i_r[j-1])$}
			\State {\bf Send: } broadcasts symbolic view $\LS^i_r[j]$
			\State {\bf Receive: } $\Pi^i_r \leftarrow \{\LS^k_r \mid 1 \leq k \leq \Monitors\}$
			\State {\bf Compute: } $\LS^i_{r+1}[0] \leftarrow \mathit{LC}(\Pi^i_r)$ \Comment{Section~\ref{subsec:union}}
		\EndFor
		\State $\Result^i \leftarrow \bigcup_{r \in [1, \lceil N/k \rceil + 1]} \LS^i_r[0]$
	\end{algorithmic}

%% file: eval.tex
\section{Case Study and Evaluation}
\label{sec:eval}

In this section, we analyze our SMT-based decentralized monitoring solution. We note that we are 
not concerned about data collections, data transfer, etc, as given a distributed setting, the 
runtime of the actual SMT encoding will be the most dominating aspect of the monitoring process. 
We evaluate our proposed solution using traces collected from synthetic experiments 
(Section~\ref{subsec:synthetic}) and case studies involving several industrial control 
systems and RACE dataset (Section~\ref{subsec:ics}). The implementation of our approach can be 
found on Google Drive(\url{https://tinyurl.com/2p6ddjnr}).

\subsection{Synthetic Experiments}
\label{subsec:synthetic}

\subsubsection{Setup}
Each experiment consists of two stages: (1) generation of the distributed stream and (2) 
verification. For data generation, we develop a synthetic program that randomly generates a 
distributed stream (i.e., the state of the local computation for a set of streams). We 
assume that streams are of the type \texttt{Float}, \texttt{Integer} or \texttt{Boolean}. For 
the streams of the type \texttt{Float} and \texttt{Integer}, the initial value is a random 
value \texttt{s[0]} and we generate the subsequent values by \texttt{s[i-1] + N(0, 2)}, for 
all $i \geq 1$. We also make sure that the value of a stream is always non-negative. On the 
other hand, for streams of the type \texttt{Boolean}, we start with either $\tru$ or $\fals$ and 
then for the subsequent values, we stay at the same value or alter using a Bernoulli 
distribution of $B(0.8)$, where a $\tru$ signifies the same value and a $\fals$ denotes a change in value.

For the monitor, we study the approach using Bernoulli distribution $B(0.2)$, $B(0.5)$ and 
$B(0.8)$ as the read distribution of the events. A higher readability offers each event to be read by 
higher number of monitors. We also make sure that each event is read by at least one monitor in 
accordance with the proposed approach. To test the approach with respect to different types of 
stream expression, we use the following arithmetic and logical expressions.

\begin{lstlisting}
input a1 : uint
input a2 : uint
output arithExp := a1 + a2
output logicExp := (a1 > 2) && (a2 < 8)
\end{lstlisting}

\subsubsection{Result - Analysis}

We study different parameters and analyze how it effects the runtime and the message size in our approach. All experiments were conducted on a 2017 MacBook Pro with 3.5GHz Dual-Core Intel core i7 processor and 16GB, 2133 MHz LPDDR3 RAM. Unless specified otherwise all experiments consider number of streams, $|\allStream| = 3$, time synchronization constant, $\epsilon_M = \epsilon = 3s$, number of monitors same as the number of streams, computation length, $N = 100$, with $k = 3$ with a read distribution $B(0.8)$.

\vspace{2mm}
\noindent{\bf Time Synchronization Constant.} Increasing the value of the time synchronization constant $\epsilon$, increases the possible number of concurrent events that needs to be considered. This increases the complexity of evaluating the \lola specification and there-by increasing the runtime of the algorithm. In addition to this, higher number of $\epsilon$ corresponds to higher number of possible streams that needs to be considered. We observe that the runtime increases exponentially with increasing the value of $\epsilon$ in Fig.~\ref{graph:runtime_epsilon}, as expected. An interesting observation is that with increasing the value of $k$, the runtime increases at a higher rate until it reaches the threshold where $k = \epsilon$. This is due to the fact, that the number of streams to be considered increases exponentially but ultimately gets bounded by the number of events present in the computation.

Increasing the value of the time synchronization constant is also directly proportional to the number of evaluated results at each instance of time. This is because, each stream corresponds to a unique value being evaluated until it gets bounded by the total number of possible evaluations, as can be seen in Fig.~\ref{graph:message_epsilon}. However, comparing Figs.~\ref{graph:runtime_epsilon} and \ref{graph:message_epsilon}, we see that the runtime increases at a faster rate to the size of the message. This owes to the fact that initially a SMT instance evaluates  unique values at all instance of time. However, as we start reaching all possible evaluations for certain instance of time, only a fraction of the total time instance evaluates to unique values. This is the reason behind the size of the message reaching its threshold faster than the runtime of the monitor.

\begin{figure*}[t]
  \centering
  \subcaptionbox{Epsilon\label{graph:runtime_epsilon}}
      {\scalebox{0.3}{\input{graph_runtime_epsilon}}}
    \hfill
    \subcaptionbox{Number of Streams\label{graph:runtime_numStream}}
      {\scalebox{0.3}{\input{graph_runtime_numStream}}}
    \hfill
    \subcaptionbox{Different \lola Specification\label{graph:runtime_formula}}
      {\scalebox{0.3}{\input{graph_runtime_diffFormula}}}
      
    \caption{Impact of different parameters on runtime for synthetic data.}

 \vspace{-2mm}
\end{figure*}
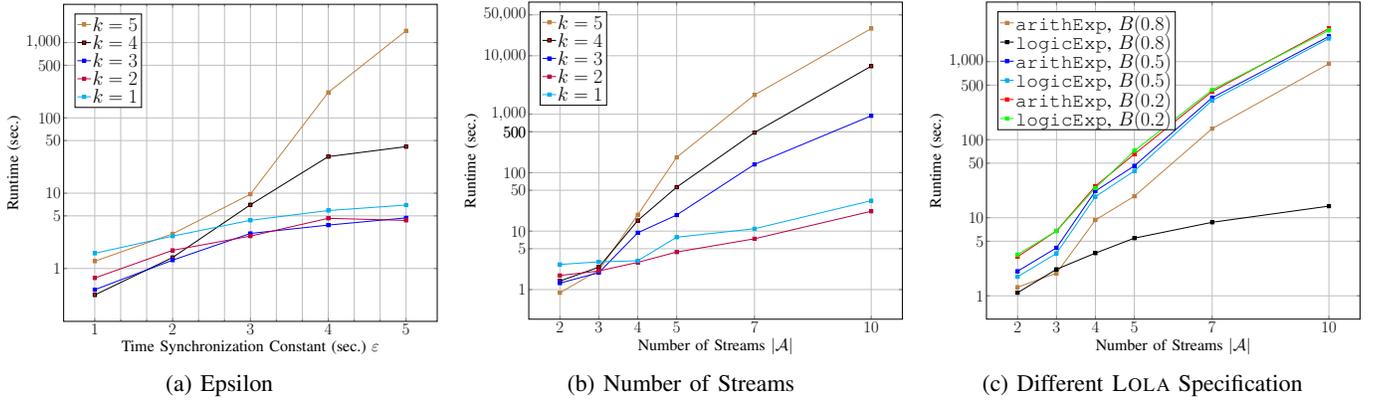

\begin{figure*}[t]
  \centering
  \subcaptionbox{Epsilon\label{graph:message_epsilon}}
      {\scalebox{0.3}{\input{graph_message_epsilon}}}
    \hfill
    \subcaptionbox{Number of Streams\label{graph:message_numStream}}
      {\scalebox{0.3}{\input{graph_message_numStream}}}
    \hfill
    \subcaptionbox{Different \lola Specification\label{graph:message_formula}}
      {\scalebox{0.3}{\input{graph_message_diffFormula}}}
      
    \caption{Impact of different parameters on message size for synthetic data.}

 \vspace{-2mm}
\end{figure*}
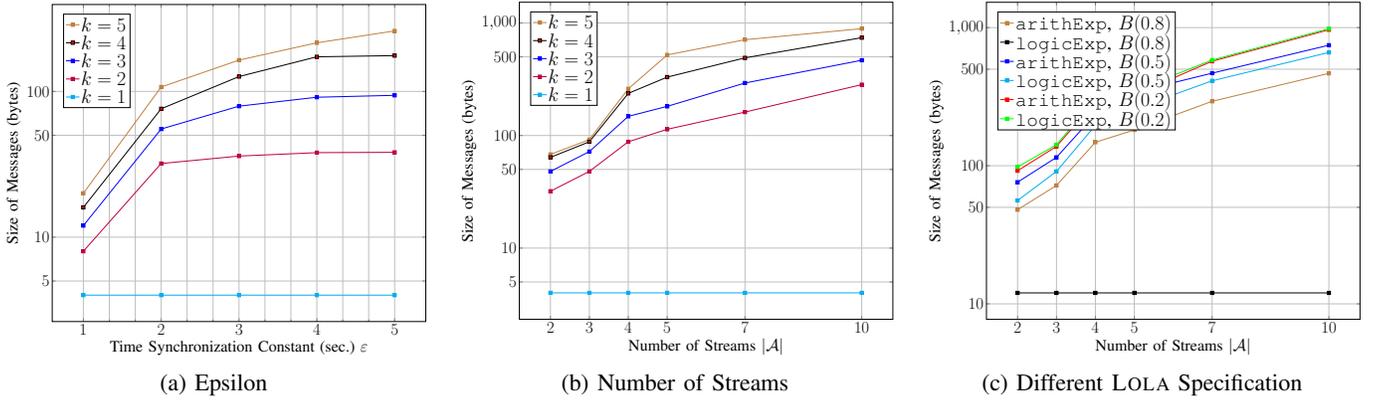

\vspace{2mm}
\noindent{\bf Type of Stream Expression.} Stream expressions can be divided into two major types, one consisting of arithmetic operations and the other involving logical operations. Arithmetic operations can evaluate to values in the order of $O(|\allStream| . \epsilon)$, where as logical operations can only evaluate to either $\tru$ or $\fals$. When the monitors have high readability of the distributed stream, it is mostly the case, that the monitor was able to evaluate the stream expression. Thus, we observe in Fig.~\ref{graph:runtime_formula} that the runtime grows exponentially for evaluating arithmetic expressions but is linear for logical expressions. However, with low readability of the computation, irrespective of the type of expression, both takes exponential time since neither can completely evaluate the stream expression. So, each monitor has to generate all possible streams.

Similarly, for high readability and logical expressions, the message size is constant given the monitor was was able to evaluate the stream expression. However with low readability, message size for evaluating logical expressions matches with that of its arithmetic counterpart. This can be seen in Fig.~\ref{graph:message_formula} and is due to the fact, that with low readability, complete evaluation of the expression is not possible at a monitor and thus needs to send the rewritten expression with the values observed to the other monitors where it will be evaluated.

\vspace{2mm}
\noindent{\bf Number of Streams.} As the number of streams increases, the number of events increase linearly and thereby making exponential increase in the number of possible synchronous streams (due to interleavings). This can be seen in Fig.~\ref{graph:runtime_numStream}, where the runtime increases exponentially with increase in the number of streams in the distributed stream. Similarly, in Fig.~\ref{graph:message_numStream}, increase in the number of streams linearly effects the number of unique values that the \lola expression can evaluate to and there-by increasing the size of the message.

\subsection{Case Studies: Decentralized ICS and Flight Control RV}
\label{subsec:ics}

We put our runtime verification approach to the test with respect to several industrial control 
system datasets that includes data generated by a (1) Secure Water Treatment plant (SWaT)~\cite{swat}, comprising of six processes, corresponding to different physical and control components; (2) a Power Distribution system~\cite{power} that includes readings from four phaser measurement unit (PMU) that measures the electric waves on an electric grid, and (3) a Gas Distribution system~\cite{gas} that includes messages to and from the PLC. In these ICS, we monitor for correctness of system properties. Additionally we monitor for mutual separation between all pairs of aircraft in RACE~\cite{race} dataset, that consists of SBS messages from aircrafts. For more details about each of the systems along with the \lola specifications refer to the Appendix~\ref{sec:ics-dataset}.

For our setting we assume, each component has its own asynchronous local clock, with varying 
time synchronization constant. Next we discuss the results of verifying different ICS with respect 
to \lola specifications. 

\paragraph*{\bf Result Analysis}

\begin{figure}
\centering
{\scalebox{0.6}{\input{graph_ics}}}
\caption{False-Positives for ICS Case-Studies}
\label{graph:ics-result}
\vspace{-7mm}
\end{figure}
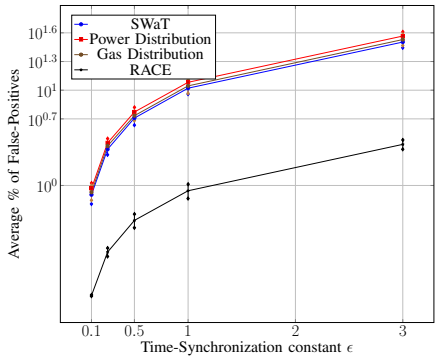

We employed same number of monitors as the number of components for each of the ICS case-studies and 
divided the entire airspace into 9 different ones with one monitor responsible for each. We observe 
that our approach does not report satisfaction of system property when there 
has been an attack on the system in reality (false-negative). However, due to the assumption of 
partial-synchrony among the components, our approach may report false positives, i.e., it reports a 
violation of the system property even when there was no 
attack on the system. As can be seen in Fig.~\ref{graph:ics-result}, with decreasing time 
synchronization constant, the number of false-positives reduce as well. This is due to the fact that 
with decreasing $\epsilon$, less events are considered to be concurrent by the monitors. This makes 
the partial-ordering of events as observed by the monitor closer to the actual-ordering of events 
taking place in the system. 

We get significantly better result for aircraft monitoring with fewer false-positives compared to 
the other dataset. This can be attributed towards Air Traffic Controllers maintaining greater 
separation between two aircrafts than the minimum that is recommended. As part of our monitoring of 
other ICS, we would like to report that our monitoring approach could successfully detect several 
attacks which includes underflow and overflow of tank and sudden change in quality of water in SWaT, 
differentiate between manual tripping of the breaker from the breaker being tripped due to a 
short-circuit in Power Distribution and Single-point data injection in Gas distribution.

%% file: graph_runtime_epsilon.tex
        \pgfplotsset{every tick label/.append style={font=\huge}}
    
    \pgfplotsset{every axis/.append style={
                       label style={font=\huge}}}
        \begin{tikzpicture}
            \begin{axis}[
                width=\textwidth,
                grid=both,
                minor tick num=2,
                xlabel={Time Synchronization Constant (sec.) $\varepsilon$},
                ylabel={Runtime (sec.)},
                legend style={nodes={scale=1, transform shape}, font=\Huge},
            y filter/.code=\pgfmathparse{#1 * 2.30258509299},
            ymode = log,
            log ticks with fixed point,
            ytick = {1, 5, 10, 50, 100, 500, 1000},
            y label style={font=\huge},
            x label style={font=\huge},
            xtick = {1, 2, 3, 4, 5},
            legend pos = north west
                ]
                \addplot[color=brown, mark=square*, ultra thin, mark options={solid, fill=brown, scale=1}] coordinates {
(1, 0.09677101653)
(2, 0.4596939765)
(3, 0.9872192299)
(4, 2.339053736)
(5, 3.157426545)
                };
                \addlegendentry{$k=5$}
                
                \addplot[color=black, mark=square*, ultra thin, mark options={solid, fill=red, scale=01}] coordinates {
(1, -0.351639989)
(2, 0.1461280357)
(3, 0.8463371121)
(4, 1.485721426)
(5, 1.618048097)
                };
                \addlegendentry{$k=4$}
                
                \addplot[color=blue, mark=square*, ultra thin, mark options={solid, fill=blue, scale=1}] coordinates {
(1, -0.2820799742)
(2, 0.1105897103)
(3, 0.4653828514)
(4, 0.5763413502)
(5, 0.6720978579)
                };
                \addlegendentry{$k=3$}
                
                \addplot[color=purple, mark=square*, ultra thin, mark options={solid, fill=purple, scale=1}] coordinates {
(1, -0.1260984021)
(2, 0.2410481507)
(3, 0.4300750556)
(4, 0.6663307443)
(5, 0.6386888867)
                };
                \addlegendentry{$k=2$}
                
                \addplot[color=cyan, mark=square*, ultra thin, mark options={solid, fill=cyan, scale=1}] coordinates {
(1, 0.2016701796)
(2, 0.42975228)
(3, 0.6396856612)
(4, 0.7696726641)
(5, 0.8422971343)
                };
                \addlegendentry{$k=1$}

            \end{axis}
        \end{tikzpicture}

%% file: graph_runtime_numStream.tex
        \pgfplotsset{every tick label/.append style={font=\huge}}
    
    \pgfplotsset{every axis/.append style={
                       label style={font=\huge}}}
        \begin{tikzpicture}
            \begin{axis}[
                width=\textwidth,
                grid=both,
                minor tick num=2,
                xlabel={Number of Streams $|\allStream|$},
                ylabel={Runtime (sec.)},
                legend style={nodes={scale=1, transform shape}, font=\Huge},
            y filter/.code=\pgfmathparse{#1 * 2.30258509299},
            ymode = log,
            log ticks with fixed point,
            ytick = {1, 5, 10, 50, 100, 500, 1000, 500, 10000, 50000},
            y label style={font=\huge},
            x label style={font=\huge},
            xtick = {2, 3, 4, 5, 7, 10},
            legend pos = north west
                ]
                \addplot[color=brown, mark=square*, ultra thin, mark options={solid, fill=brown, scale=1}] coordinates {
(2, -0.05453141487)
(3, 0.3304137733)
(4, 1.278753601)
(5, 2.26245109)
(7, 3.329397879)
(10, 4.46505552)
                };
                \addlegendentry{$k=5$}
                
                \addplot[color=black, mark=square*, ultra thin, mark options={solid, fill=red, scale=01}] coordinates {
(2, 0.1461280357)
(3, 0.3830969299)
(4, 1.1811287)
(5, 1.752048448)
(7,  2.685741739)
(10, 3.8213825)
                };
                \addlegendentry{$k=4$}
                
                \addplot[color=blue, mark=square*, ultra thin, mark options={solid, fill=blue, scale=1}] coordinates {
(2, 0.1105897103)
(3, 0.2893659515)
(4, 0.9712758487)
(5, 1.275219187)
(7, 2.141659815)
(10, 2.972764095)
                };
                \addlegendentry{$k=3$}
                
                \addplot[color=purple, mark=square*, ultra thin, mark options={solid, fill=purple, scale=1}] coordinates {
(2, 0.2410481507)
(3, 0.3184807252)
(4, 0.4656802116)
(5, 0.6421676344)
(7, 0.8688794462)
(10, 1.339053736)
                };
                \addlegendentry{$k=2$}
                
                \addplot[color=cyan, mark=square*, ultra thin, mark options={solid, fill=cyan, scale=1}] coordinates {
(2, 0.42975228)
(3, 0.4747988188)
(4, 0.4930395883)
(5, 0.8955330395)
(7, 1.040523226)
(10, 1.520352504)
                };
                \addlegendentry{$k=1$}

            \end{axis}
        \end{tikzpicture}

%% file: graph_runtime_diffFormula.tex
        \pgfplotsset{every tick label/.append style={font=\huge}}
    
    \pgfplotsset{every axis/.append style={
                       label style={font=\huge}}}
        \begin{tikzpicture}
            \begin{axis}[
                width=\textwidth,
                grid=both,
                minor tick num=2,
                xlabel={Number of Streams $|\allStream|$},
                ylabel={Runtime (sec.)},
                legend style={nodes={scale=1, transform shape}, font=\Huge},
            y filter/.code=\pgfmathparse{#1 * 2.30258509299},
            ymode = log,
            log ticks with fixed point,
            ytick = {1, 5, 10, 50, 100, 500, 1000},
            y label style={font=\huge},
            x label style={font=\huge},
            xtick = {1, 2, 3, 4, 5, 7, 10},
            legend pos = north west
                ]
                \addplot[color=brown, mark=square*, ultra thin, mark options={solid, fill=brown, scale=1}] coordinates {
(2, 0.1105897103)
(3, 0.2893659515)
(4, 0.9712758487)
(5, 1.275219187)
(7, 2.141659815)
(10, 2.972764095)
                };
                \addlegendentry{\texttt{arithExp}, $B(0.8)$}
                
                \addplot[color=black, mark=square*, ultra thin, mark options={solid, fill=black, scale=01}] coordinates {
(2, 0.03925544381)
(3, 0.3394514413)
(4, 0.5467893516)
(5, 0.7382254481)
(7, 0.9415611202)
(10, 1.148849343)
                };
                \addlegendentry{\texttt{logicExp}, $B(0.8)$}

                \addplot[color=blue, mark=square*, ultra thin, mark options={solid, fill=blue, scale=1}] coordinates {
(2, 0.314709693)
(3, 0.6143698395)
(4, 1.341948646)
(5, 1.664199936)
(7, 2.537441283)
(10, 3.322777315)
                };
                \addlegendentry{\texttt{arithExp}, $B(0.5)$}
                
                \addplot[color=cyan, mark=square*, ultra thin, mark options={solid, fill=cyan, scale=01}] coordinates {
(2, 0.2432861461)
(3, 0.543571424)
(4, 1.270329397)
(5, 1.599337133)
(7, 2.502399805)
(10, 3.296344835)
                };
                \addlegendentry{\texttt{logicExp}, $B(0.5)$}

                \addplot[color=red, mark=square*, ultra thin, mark options={solid, fill=red, scale=1}] coordinates {
(2, 0.4996870826)
(3, 0.8317418336)
(4, 1.405517107)
(5, 1.815444916)
(7, 2.617377824)
(10, 3.424358863)
                };
                \addlegendentry{\texttt{arithExp}, $B(0.2)$}
                
                \addplot[color=green, mark=square*, ultra thin, mark options={solid, fill=green, scale=01}] coordinates {
(2, 0.5287110684)
(3, 0.8286598965)
(4, 1.381656483)
(5, 1.856668484)
(7, 2.63844932)
(10, 3.400771047)
                };
                \addlegendentry{\texttt{logicExp}, $B(0.2)$}

            \end{axis}
        \end{tikzpicture}

%% file: graph_message_epsilon.tex
        \pgfplotsset{every tick label/.append style={font=\huge}}
    
    \pgfplotsset{every axis/.append style={
                       label style={font=\huge}}}
        \begin{tikzpicture}
            \begin{axis}[
                width=\textwidth,
                grid=both,
                minor tick num=2,
                xlabel={Time Synchronization Constant (sec.) $\varepsilon$},
                ylabel={Size of Messages (bytes)},
                legend style={nodes={scale=1, transform shape}, font=\Huge},
            ymode = log,
            log ticks with fixed point,
            ytick = {1, 5, 10, 50, 100, 500},
            y label style={font=\huge},
            x label style={font=\huge},
            xtick = {1, 2, 3, 4, 5},
            legend pos = north west
                ]
                \addplot[color=brown, mark=square*, ultra thin, mark options={solid, fill=brown, scale=1}] coordinates {
(1, 20)
(2, 107.2)
(3, 164)
(4, 215.5)
(5, 260)
                };
                \addlegendentry{$k=5$}
                
                \addplot[color=black, mark=square*, ultra thin, mark options={solid, fill=red, scale=01}] coordinates {
(1, 16)
(2, 76)
(3, 126.4)
(4, 172.8)
(5, 176)
                };
                \addlegendentry{$k=4$}
                
                \addplot[color=blue, mark=square*, ultra thin, mark options={solid, fill=blue, scale=1}] coordinates {
(1, 12)
(2, 55.2)
(3, 79.2)
(4, 91.2)
(5, 94)
                };
                \addlegendentry{$k=3$}
                
                \addplot[color=purple, mark=square*, ultra thin, mark options={solid, fill=purple, scale=1}] coordinates {
(1, 8)
(2, 32)
(3, 36)
(4, 38)
(5, 38.2)
                };
                \addlegendentry{$k=2$}
                
                \addplot[color=cyan, mark=square*, ultra thin, mark options={solid, fill=cyan, scale=1}] coordinates {
(1, 4)
(2, 4)
(3, 4)
(4, 4)
(5, 4)
                };
                \addlegendentry{$k=1$}

            \end{axis}
        \end{tikzpicture}

%% file: graph_message_numStream.tex
        \pgfplotsset{every tick label/.append style={font=\huge}}
    
    \pgfplotsset{every axis/.append style={
                       label style={font=\huge}}}
        \begin{tikzpicture}
            \begin{axis}[
                width=\textwidth,
                grid=both,
                minor tick num=2,
                xlabel={Number of Streams $|\allStream|$},
                ylabel={Size of Messages (bytes)},
                legend style={nodes={scale=1, transform shape}, font=\Huge},
            ymode = log,
            log ticks with fixed point,
            ytick = {1, 5, 10, 50, 100, 500, 1000},
            y label style={font=\huge},
            x label style={font=\huge},
            xtick = {2, 3, 4, 5, 7, 10},
            legend pos = north west
                ]
                \addplot[color=brown, mark=square*, ultra thin, mark options={solid, fill=brown, scale=1}] coordinates {
(2, 68)
(3, 92)
(4, 260)
(5, 520)
(7, 710)
(10, 890)
                };
                \addlegendentry{$k=5$}
                
                \addplot[color=black, mark=square*, ultra thin, mark options={solid, fill=red, scale=01}] coordinates {
(2, 64)
(3, 88)
(4, 236.8)
(5, 330.4)
(7, 490)
(10, 740)
                };
                \addlegendentry{$k=4$}
                
                \addplot[color=blue, mark=square*, ultra thin, mark options={solid, fill=blue, scale=1}] coordinates {
(2, 48)
(3, 72)
(4, 148)
(5, 181.6)
(7, 292.8)
(10, 466.4)
                };
                \addlegendentry{$k=3$}
                
                \addplot[color=purple, mark=square*, ultra thin, mark options={solid, fill=purple, scale=1}] coordinates {
(2, 32)
(3, 48)
(4, 88)
(5, 113.6)
(7, 161.6)
(10, 283.2)
                };
                \addlegendentry{$k=2$}
                
                \addplot[color=cyan, mark=square*, ultra thin, mark options={solid, fill=cyan, scale=1}] coordinates {
(2, 4)
(3, 4)
(4, 4)
(5, 4)
(7, 4)
(10, 4)
                };
                \addlegendentry{$k=1$}

            \end{axis}
        \end{tikzpicture}

%% file: graph_message_diffFormula.tex
        \pgfplotsset{every tick label/.append style={font=\huge}}
    
    \pgfplotsset{every axis/.append style={
                       label style={font=\huge}}}
        \begin{tikzpicture}
            \begin{axis}[
                width=\textwidth,
                grid=both,
                minor tick num=2,
                xlabel={Number of Streams $|\allStream|$},
                ylabel={Size of Messages (bytes)},
                legend style={nodes={scale=1, transform shape}, font=\Huge},
            ymode = log,
            log ticks with fixed point,
            ytick = {1, 5, 10, 50, 100, 500, 1000},
            y label style={font=\huge},
            x label style={font=\huge},
            xtick = {1, 2, 3, 4, 5, 7, 10},
            legend pos = north west
                ]
                \addplot[color=brown, mark=square*, ultra thin, mark options={solid, fill=brown, scale=1}] coordinates {
(2, 48)
(3, 72)
(4, 148)
(5, 181.6)
(7, 292.8)
(10, 466.4)
                };
                \addlegendentry{\texttt{arithExp}, $B(0.8)$}
                
                \addplot[color=black, mark=square*, ultra thin, mark options={solid, fill=black, scale=01}] coordinates {
(2, 12)
(3, 12)
(4, 12)
(5, 12)
(7, 12)
(10, 12)
                };
                \addlegendentry{\texttt{logicExp}, $B(0.8)$}

                \addplot[color=blue, mark=square*, ultra thin, mark options={solid, fill=blue, scale=1}] coordinates {
(2, 76)
(3, 115)
(4, 236)
(4, 290)
(7, 468)
(10, 746)
                };
                \addlegendentry{\texttt{arithExp}, $B(0.5)$}
                
                \addplot[color=cyan, mark=square*, ultra thin, mark options={solid, fill=cyan, scale=01}] coordinates {
(2, 56)
(3, 91)
(4, 196)
(5, 258)
(7,  412)
(10, 661)
                };
                \addlegendentry{\texttt{logicExp}, $B(0.5)$}

                \addplot[color=red, mark=square*, ultra thin, mark options={solid, fill=red, scale=1}] coordinates {
(2, 92)
(3, 138)
(4, 292)
(5, 326)
(7, 571)
(10, 965)
                };
                \addlegendentry{\texttt{arithExp}, $B(0.2)$}
                
                \addplot[color=green, mark=square*, ultra thin, mark options={solid, fill=green, scale=01}] coordinates {
(2, 98)
(3, 142)
(4, 304)
(5, 349)
(7, 582)
(10, 982)
                };
                \addlegendentry{\texttt{logicExp}, $B(0.2)$}

            \end{axis}
        \end{tikzpicture}

%% file: graph_ics.tex
\pgfplotsset{every tick label/.append style={font=\huge}}
    
    \pgfplotsset{every axis/.append style={
                       label style={font=\huge}}}

\begin{tikzpicture}[scale=0.5]
\begin{axis}[
width = \textwidth, 
grid = both, 
legend pos = north west, 
ymode = log, 
xlabel = {Time-Synchronization constant $\epsilon$}, 
ylabel={Average \% of False-Positives}, 
xtick = {0.1, 0.5, 1, 2, 3}, 
ytick = {1, 5, 10, 20, 40}, 
legend style={font=\huge}, 
bar width=0.2]

\addplot+[error bars/.cd, y dir=both, y explicit, error mark=diamond*,  error bar style={color=blue}
]
coordinates{
    (3, 32)    +-  (0, 4.1)
    (1, 10.5)    +-  (0, 1.4)
    (0.5, 5.1)    +-    (0, 0.8)
    (0.25, 2.4)   +-   (0, 0.3)
    (0.1, 0.8)    +-    (0, 0.16)
};
\addlegendentry{SWaT}

\addplot+[error bars/.cd, y dir=both, y explicit, error mark=diamond*,  error bar style={color=red}
]
coordinates{
    (3, 37)    +-  (0, 3.9)
    (1, 12.1)    +-  (0, 1.2)
    (0.5, 5.9)    +-    (0, 0.7)
    (0.25, 2.8)   +-   (0, 0.3)
    (0.1, 0.93)    +-    (0, 0.13)
};
\addlegendentry{Power Distribution}

\addplot+[error bars/.cd, y dir=both, y explicit, error mark=diamond*,  error bar style={color=brown}
]
coordinates{
    (3, 34)    +-  (0, 4.5)
    (1, 11.1)    +-  (0, 1.8)
    (0.5, 5.4)    +-    (0, 0.6)
    (0.25, 2.6)   +-   (0, 0.2)
    (0.1, 0.85)    +-    (0, 0.15)
};
\addlegendentry{Gas Distribution}

\addplot+[error bars/.cd, y dir=both, y explicit, error mark=diamond*,  error bar style={color=black}
]
coordinates{
	(3, 2.7)    +-	(0, 0.3)
    (1, 0.88)    +-	(0, 0.15)
    (0.5, 0.43)    +-	(0, 0.07)
    (0.25, 0.2)   +-	(0, 0.02)
    (0.1, 0.07)    +-	(0, 0.001)
};
\addlegendentry{RACE}

\end{axis}

\end{tikzpicture}

%% file: appendix.tex
\section{Appendix}
\label{sec:appendix}

\subsection{\lola Syntax}
\label{sec:moreLOLA}

A stream expression is constructed as follows:
\begin{itemize}
    \item If $c$ is a constant of type $\type$, then $c$ is an atomic stream expression of 
    type $\type$
    \item If $s$ is a stream variable of type $\type$, then $s$ is an atomic stream expression 
    of type $\type$.
    \item If $f:\type_1\times \type_2 \times \cdots \type_k \rightarrow \type$ is a k-ary 
    operator 
    and for $1 \leq i \leq k$, $e_i$ is an expression of type $\type_i$, then 
    $f(e_1, e_2, \cdots, e_k)$ is a stream expression of type $\type$
    \item If $b$ is a stream expression of type \texttt{boolean} and $e_1, e_2$ are stream 
    expressions of type $\type$, then $\ite(b, e_1, e_2)$ is a stream expression of type 
    $\type$, where $\ite$ is the abbreviated form of \textit{if-then-else}.
    \item If $e$ is a stream expression of type $\type$, $c$ is a constant of type $\type$ and 
    $i$ is an integer, then $e[i, c]$ is a stream expression of type $\type$. $e[i, c]$ refers 
    to the value of the expression $e$ offset by $i$ positions from the current position. In 
    case the offset takes it beyond the end or before the beginning of the stream, then the 
    {\em default} value is $c$.
\end{itemize}

Furthermore, \lola can be used to compute {\em incremental statistics}, where a given a stream, 
$\stream$, a function, $f_{\stream}(v, u)$, computes a measure, where $u$ represents the 
measure thus far and $v$, the current value. Given a sequence of values, $v_1, v_2, \cdots, v_n$, 
with a default value $d$, the measure over the data is given as
$$u = f_{\stream}(v_n, f_{\stream}(v_{n-1}, \cdots, f_{\alpha}(v_1, d)))$$
Example of such functions include \textit{count}, $f_{\mathit{count}}(v, u) = u+1$, \textit{sum}, 
$f_{\mathit{sum}}(v, u) = u+v$, \textit{max}, $f_{\mathit{max}}(v, u) = \max\{v, u\}$, among 
others. Aggregate functions like \textit{average}, can be defined using two incremental functions, 
\textit{count} and \textit{sum}.

\subsection{Proofs}
\label{sec:proofs}

\begin{lemma}
Let $\allStream = \{S_1, S_2, \cdots, S_n \}$ be a distributed system and $\varphi$ be an \lola 
specification. Algorithm~\ref{alg:monitor} terminates when monitoring a terminating distributed 
system.
\end{lemma}

\begin{proof}
First, we note that our algorithm is designed for terminating system, also, note that a terminating 
program only produces a finite distributed computation. In order to prove the lemma, let us assume 
that the system send out a \textit{stop} signal to all monitor processes when it terminates. When 
such a signal is received by a monitor, it starts evaluating the output stream expression using the 
terminal associated equations. This might arise to two cases. One where all the values required for 
the evaluation has been observed or one where the values required for the evaluation has not been 
observed. Although the termination of the monitor process for the first case is trivial, the 
termination of the monitor process for the second case is dependent upon replacing such unobserved 
stream value by the default value of the stream expression. Thus, terminating the monitor process 
eventually.
\end{proof}









\begin{theorem}
Algorithm~\ref{alg:monitorUpdt} solves the problem stated in Section~\ref{sec:problem}.
\end{theorem}

\begin{proof}
We prove the soundness and correctness of Algorithm~\ref{alg:monitorUpdt}, by dividing it into three 
steps. In the first step we prove that given a \lola specification, $\varphi$, the values of the 
output stream when computed over the distributed computation, $(\Events, \hb)$, of length $N$ is 
the same as when the distributed computation is divided into $\frac{N}{k}$ computation rounds of 
length $k$ each. Second, we prove that for all time instances the stream equation is eventually 
evaluated after the communication round. Finally we prove the set of all evaluated result is 
consistent over all monitors in the system.

{\bf Step 1: } From our approach, we see that the value of a output stream variable, is evaluated 
on the events present in the consistent cut with time $j$. Therefore, we can reduce the proof to:
$$\Sr(\Events, \hb) = \Sr(\Events_1.\Events_2 \cdots \Events_{\frac{N}{k}}, \hb)$$
\begin{itemize}
	\item $(\Rightarrow)$ Let $\cc_k$ be a consistent cut such that $\cc_k$ is in $\Sr(\Events, \hb)$
	, but not in $\Sr(\Events_1.\Events_2 \cdots \Events_{\frac{N}{k}}, \hb)$, for some 
	$k \in [0, |\Events|]$. This implies that the frontier of $\cc_k$, 
	$\front(\cc_k) \not\subseteq \Events_1$ and $\front(\cc_k) \not\subseteq \Events_2$ and 
	$\cdots$ and $\front(\cc_k) \not\subseteq \Events_{\frac{N}{k}}$. However, this is not 
	possible, as according to the computation round construction in Section~\ref{subsec:final}, 
	there must be a $\Events_i$, where $1 \leq i \leq {\frac{N}{k}}$ such that 
	$\front(\cc_k) \subseteq \Events_i$. Therefore, such $\cc_k$ cannot exist, and 
	$(\stream_1, \stream_2, \cdots, \stream_n) \in \Sr(\Events, \hb) \implies (\stream_1, \stream_2, \cdots, \stream_n) \in \Sr(\Events_1.\Events_2 \cdots \Events_{\frac{N}{k}}, \hb)$.

	\item $(\Leftarrow)$ Let $\cc_k$ be a consistent cut such that $\cc_k$ is in 
	$\Sr(\Events_1.\Events_2 \cdots \Events_{\frac{N}{k}}, \hb)$ but not in $\Sr(\Events, \hb)$ 
	for some $k \in [0, |\Events|]$. This implies, $\front(\cc_k) \subseteq \Events_i$ and 
	$\front(\cc_k) \not\subseteq \Events$ for some $i \in [1, \frac{N}{k}]$. However, this is not 
	possible due to the fact that $\forall i \in [1, \frac{N}{k}]. \Events_i \subset \Events$. 
	There, such $\cc_k$ cannot exist, and 
	$(\stream_1, \stream_2, \cdots, \stream_n) \in \Sr(\Events_1.\Events_2 \cdots \Events_{\frac{N}{k}}, \hb) \implies (\stream_1, \stream_2, \cdots, \stream_n) \in \Sr(\Events, \hb)$.
\end{itemize}

Therefore, $\Sr(\Events, \hb) = \Sr(\Events_1.\Events_2 \cdots \Events_{\frac{N}{k}}, \hb)$.

{\bf Step 2: } Given a output stream expression $s_i$ and the dependency graph 
$G = \langle V, E \rangle$, for each $\langle s_i, s_k, w \rangle \in E$, evaluating the value at 
time instance $j \in [1, N]$, $\stream_k(j+w) \neq \natural$ or $\stream_k(j+w) = \natural$ or 
$\stream_k(w+j)$ not observed.
\begin{itemize}
	\item If $\stream_k(j+w) \neq \natural$, then we evaluate the stream expression

	\item If $\stream_k(j+w) = \natural$, there exists at-least one other monitor where 
	$\stream_k(j+w) \neq \natural$. Thereby evaluating the stream expression, followed by sharing 
	the the evaluated result with all other monitors

	\item If $\stream_k(w+j)$ not observed, then at some future evaluation round and at some monitor
	$\stream_k(j+w) \neq \natural$ and there-by evaluating the stream expression $s_i$
\end{itemize}

Similarly, it can be proved for $\langle s_i, t_k, w \rangle \in E$.

{\bf Step 3: } Each monitor in our approach is fault-proof with communication taking place 
between all pairs of monitors. We also assume, all messages are eventually received by the monitors. 
This guarantees all observations are either directly or indirectly read by each monitor.

Together with Step 1 and 2, soundness and correctness of Algorithm~\ref{alg:monitor} is proved.
\end{proof}

\begin{theorem}
Let $\varphi$ be a \lola specification and $(\Events, \hb)$ be a distributed 
stream consisting of $|\allStream|$ streams. The message complexity of Algorithm~\ref{alg:monitorUpdt} with $|\Monitors|$ monitors is
$$O\big(\epsilon^{|\allStream|} N |\Monitors|^2\big) \;\;\; \Omega(N |\Monitors|^2)$$
\end{theorem}

\begin{proof}

We analyze the complexity of each part of Algorithm~\ref{alg:monitorUpdt}. The algorithm has a 
nested loop. The outer loop iterates for $\lceil N/k \rceil$ times, that is $O(N)$. The inner loop 
is dependent on the number of unique evaluations of the stream expression.

\begin{itemize}
	\item {\bf Upper-bound} Due to our assumption of partial-synchrony, each event's time of 
	occurrence can be off by $\epsilon$. This makes the maximum number of unique evaluations 
	in the order of $O(\epsilon^{|\allStream|})$.

	\item {\bf Lower-bound} The minimum number of unique evaluations is in the order of $\Omega(1)$.
\end{itemize}

In the communication phase, each monitor sends $|\Monitors|$ messages to all other monitors and 
receives $|\Monitors|$ messages from all other monitors. That is $|\Monitors|^2$. Hence the message 
complexity is
$$O\big(\epsilon^{|\allStream|} N |\Monitors|^2\big) \;\;\; \Omega(N |\Monitors|^2)$$
As a side note, we would like to mention that in case of high readability of the monitors and 
evaluation of logical expression, the complexity is closer to the lower-bound, whereas with low 
readability and arithmetic expressions, the complexity is closer to the upper bound.

\end{proof}

\subsection{Industrial Control Systems}
\label{sec:ics-dataset}

\paragraph{SWaT Dataset} Secure Water Treatment (SWaT)~\cite{swat} utilizes a fully operational scaled down water 
treatment plant with a small footprint, producing 5 gallons/minute of doubly filtered water. It 
comprises of six main processes corresponding to the physical and control components of the water 
treatment facility. It starts from process P1 where it takes raw water and stores it in a tank. 
It is then passed through the pre-treatment process, P2, where the quality of the water is assessed 
and maintained through chemical dosing. The water then reaches P3 where undesirable materials are 
removed using fine filtration membranes. Any remaining chlorine is destroyed in the dechlorination 
process in P4 and the water is then pumped into the Reverse Osmosis system (P5) to reduce inorganic 
impurities. Finally in P6, water from the RO system is stored ready for distribution.

The dataset classifies different attack on the system into four types, based on the point and stage 
of the attack: Single Stage-Single Point, Single Stage-Multi Point, Multi Stage-Single Point and 
Multi Stage-Multi Point. We for the scope of this paper are the most interested in the attacks 
either covering multiple stages or multiple points. Few of the \lola specifications used are listed 
below.

\begin{lstlisting}
input FIT-101 : uint
input MV-101 : bool
input LIT-101 : uint
input P-101 : bool
input FIT-201 : uint
output inflowCorr := ite(MV-101 == true, FIT-101 > 0, FIT-101 == 0)
output outflowCorr := ite(P-101 == true, FIT-201 > 0, FIT-201 == 0)
output tankCorr := ite(MV-101 == true || P-101 == true, LIT-101 = LIT-101[-1, 0] + FIT-101[-1, 0] - FIT-201[-1, 0])
\end{lstlisting}

where \texttt{FIT-101} is the flow meter, measuring inflow into raw water tank, \texttt{MV-101} is a 
motorized valve that controls water flow to the raw water tank, \texttt{LIT-101} is the level 
transmitter of the raw water tank, \texttt{P-101} is a pump that pumps water from raw water tank to 
the second stage and \texttt{FIT-201} is the flow transmitter for the control dosing pumps. The 
above \lola specification checks the correctness of the inflow meter and valve pair (resp. outflow 
meter and pump pair) in \texttt{inflowCorr} (resp. \texttt{outflowCorr}) output expressions. On 
the other hand, \texttt{tankCorr} checks if the water level in the tank adds up to the in-flow 
and out-flow meters.

\begin{lstlisting}
input AIT-201 : uint
input AIT-202 : uint
input AIT-203 : uint
output numObv := numObv[-1, 0] + 1
output NaClAvg := (NaClAvg[-1, 0] * numObv[-1, 0] + AIT-201) / numObv
output HClAvg := (HClAvg[-1, 0] * numObv[-1, 0] + AIT-202) / numObv
output NaOClAvg := (NaOClAvg[-1, 0] * numObv[-1, 0] + AIT-202) / numObv
\end{lstlisting}

where \texttt{AIT-201}, \texttt{AIT-202} and \texttt{AIT-203} represents the NaCl, HCl and NaOCl 
levels in water respectively and \texttt{NaClAvg}, \texttt{HClAvg} and \texttt{NaOClAvg} keeps 
a track of the average levels of the corresponding chemicals in the water, where as \texttt{numObv} 
keeps a track of the total number of observations read by the monitor.

\paragraph{Power System Attack Dataset} Power System Attack Dataset~\cite{power} consists of three datasets 
developed by Mississippi State University and Oak Ridge National Laboratory. It consists of readings 
from four phaser measurement unit (PMU) or synchrophasor that measures the electric waves on an 
electric grid. Each PMU measures 29 features consisting of voltage phase angle, voltage phase 
magnitude, current phase angle, current phase magnitude for Phase A-C, Pos., Neg. and Zero. It 
also measures the frequency for relays, the frequency delta for relay, status flag for relays, etc. 
Apart from these 116 PMU measurements, the dataset also consists of 12 control panel logs, snort 
alerts and relay logs of the 4 PMU.

The dataset classifies into either natural event/no event or an attack event. Few of the \lola 
specifications used are listed below. The first attempts to detect a single-line-to-ground (1LG) 
fault.

\begin{lstlisting}
input R1-I : float
input R2-I : float
input R1-Relay : bool
input R2-Relay : bool
output R1-I-low := R1-I < 200
output R1-I-high := R1-I > 1000
output R2-I-low := R2-I < 200
output R2-I-high := R2-I > 1000
output 1LG := R1-I-high && R2-I-high && R1-Relay[+2, false] && R2-Relay[+2, false] && R1-I-low[+4, false] && R2-I-low[+4, false]
\end{lstlisting}

where \texttt{R1-I} and \texttt{R2-I} represents the current measured at the R1 and R2 PMU 
respectively. Additionally, \texttt{R1-Relay} and \texttt{R2-Relay} keeps a track of the state 
of the corresponding relay. As a part of the 1LG attack detection, we first categorize the 
current measured as either low or high depending upon the amount of the current measured. We 
categorize an attack as 1LG if both R1 and R2 detects high current flowing followed by the relay 
tripping followed by low current.

\begin{lstlisting}
input R1-PA1-I : float
input R1-PA2-I : float
input R1-PA3-I : float
output phaseBal := (R1-PA1-I - R1-PA2-I) <= 10 && (R1-PA2-I - R1-PA3-I) <= 10 && (R1-PA3-I - R1-PA1-I) <= 10
\end{lstlisting}

where \texttt{R1-PA1-I}, \texttt{R1-PA2-I} and \texttt{R1-PA3-I} are the amount of current 
measured by R1 PMU at Phase A, B and C respectively. The monitor helps us to check if the load 
on three phases are equally balanced.

\paragraph{Gas Distribution System} Gas Distributed System~\cite{gas} is a collection of labeled Remote 
Terminal Unit (RTU) telemetry streams from a Gas pipeline system in Mississippi State 
University's Critical Infrastructure Protection Center with collaboration from Oak Ridge 
National Laboratory. The telemetry streams includes messages to and from the Programmable Logic 
Controller (PLC) under normal operations and attacks involving command injection and data 
injection attack. The feature set includes the pipeline pressure, setpoint value, command data 
from the PLC, response to the PLC and the state of the solenoid, pump and the Remote Terminal Unit 
(RTU) auto-control.

One of the most common data injection attack is \textit{Fast Change}. Here the reported pipeline 
pressure value is successively varied to create a lack of confidence in the correct operation of the 
system. The corresponding \lola specification monitoring against such attack is mentioned below:

\begin{lstlisting}
input PipePress : float
input response : bool
output fastChange := ite(response, mod(PipePress - PipePress[-1, 1000]) <= 10, true)
\end{lstlisting}

where \texttt{PipePress} records the measured pipeline pressure and \texttt{response} is a flag 
variable signifying a message to the PLC. Here we consider the default pressure is 1000 psi and 
the permitted pressure change per unit time is 10 psi (these can be changed according to the 
demands of the system). Similarly we have \lola specifications monitoring other data injection 
attacks such as \textit{Value Wave Injection}, \textit{Setpoint Value Injection}, \textit{Single 
Data Injection}, etc. and command injection attacks such as \textit{Illegal Setpoint}, 
\textit{Illegal PID Command}, etc.

\paragraph{RACE Dataset} Runtime for Airspace Concept Evaluation (RACE)~\cite{race} is a framework developed by NASA that 
is used to build an event based, reactive airspace simulation. We use a dataset developed using 
this RACE framework. This dataset contains three sets of data collected on three different days. 
Each set was recorded at around 37 N Latitude and 121 W Longitude. The dataset includes all 8 
types of messages being sent by the SBS unit by using a Telnet application to listen to port 30003, 
but we only use the messages with ID `MSG 3' which is the Airborne Position Message and includes a 
flight’s latitude, longitude and altitude using which we verify the mutual separation of all pairs 
of aircraft. Furthermore, calculating the distance 
between two coordinates is computationally expensive, as we need to factor in parameters such as 
curvature of the earth. In order to speed up distance related calculations, we consider
a constant latitude distance of 111.2km and longitude distance of 87.62km, at the cost of a
negligible error margin. The corresponding \lola specification is mentioned below:

\begin{lstlisting}
input flight1_alt : float
input flight1_lat : float
input flight1_lon : float
input flight2_alt : float
input flight2_lat : float
input flight2_lon : float
output distDiff := sqrt(pow(flight1_alt - flight2_alt, 2) + pow((flight1_lon - flight2_lon)*87620, 2) + pow((flight1_lat - flight2_lat)*111200, 2))
output check := distDiff > 500
\end{lstlisting}